\documentclass[11pt]{article}
\usepackage{times}
\usepackage{amsmath,amsfonts}
\usepackage{epsfig,amssymb,amstext,xspace}
\usepackage{theorem} 
\usepackage{multirow}
\usepackage{graphicx,color}

\newcommand{\np}{{\em NP}\xspace}
\newcommand{\nphard}{\np-hard\xspace} 
\newcommand{\npcomplete}{\np-complete\xspace}

\newtheorem{theorem}{Theorem}[section]
\newtheorem{lemma}[theorem]{Lemma}

\newtheorem{claim}[theorem]{Claim}

\newtheorem{corollary}[theorem]{Corollary}

{\theorembodyfont{\rmfamily} \newtheorem{remark}[theorem]{Remark}
}

\def\blksquare{\rule{2mm}{2mm}}
\def\qedsymbol{\blksquare}
\newcommand{\bg}[1]{\medskip\noindent{\bf #1}}
\newcommand{\ed}{{\hfill\qedsymbol}\medskip}
\newenvironment{proof}{\bg{Proof : }}{\ed}
\newenvironment{proofnobox}{\bg{Proof : }}{\medskip}

\newcommand{\R}{\ensuremath{\mathbb R}}
\newcommand{\Z}{\ensuremath{\mathbb Z}}

\newcommand{\I}{\ensuremath{\mathcal I}}
\newcommand{\F}{\ensuremath{\mathcal F}}
\newcommand{\D}{\ensuremath{\mathcal D}}

\newcommand{\Kc}{\ensuremath{\mathcal K}}
\newcommand{\Lc}{\ensuremath{\mathcal L}}
\newcommand{\Oc}{\ensuremath{\mathcal O}}
\newcommand{\Pc}{\ensuremath{\mathcal P}}
\newcommand{\Qc}{\ensuremath{\mathcal Q}}
\newcommand{\Rc}{\ensuremath{\mathcal R}}

\newcommand{\OPT}{\ensuremath{\mathit{OPT}}}
\newcommand{\sm}{\ensuremath{\setminus}}
\newcommand{\es}{\ensuremath{\emptyset}}

\newcommand{\e}{\ensuremath{\epsilon}}
\newcommand{\gm}{\ensuremath{\gamma}}
\newcommand{\ld}{\ensuremath{\lambda}}

\newcommand{\junk}[1]{}

\newcommand{\sse}{\subseteq}

\newcommand{\nbr}{\ensuremath{\mathsf{nbr}}}
\newcommand{\ctr}{\ensuremath{\mathsf{ctr}}}
\newcommand{\bC}{\ensuremath{\bar C}}
\newcommand{\hC}{\ensuremath{\hat C}}

\newcommand{\assign}{\ensuremath{\leftarrow}}
\newcommand{\tx}{\ensuremath{\tilde x}}
\newcommand{\ty}{\ensuremath{\tilde y}}
\newcommand{\tz}{\ensuremath{\tilde z}}
\newcommand{\hi}{\ensuremath{\hat i}}
\newcommand{\hx}{\ensuremath{\hat x}}
\newcommand{\hy}{\ensuremath{\hat y}}
\newcommand{\hz}{\ensuremath{\hat z}}
\newcommand{\lp}{\ensuremath{\mathsf{LP}}}
\newcommand{\sg}{\ensuremath{\sigma}}

\newcommand{\tmmed}{{\small \textsf{2MMed}}\xspace}
\newcommand{\lmmed}{{\small \textsf{LCMMed}}\xspace}
\newcommand{\mmst}{\ensuremath{\mathsf{MST}}\xspace}
\newcommand{\ufl}{{\small \textsf{UFL}}\xspace}
\newcommand{\mlufl}{{\small \textsf{MLUFL}}\xspace}
\newcommand{\lb}{\ensuremath{\mathit{lb}}}
\newcommand{\lbo}{\ensuremath{\mathit{lb1}}}
\newcommand{\lbt}{\ensuremath{\mathit{lb2}}}
\newcommand{\ub}{\ensuremath{\mathit{ub}}}
\newcommand{\ubo}{\ensuremath{\mathit{ub1}}}
\newcommand{\ubt}{\ensuremath{\mathit{ub2}}}
\newcommand{\fopt}{\ensuremath{f^{\mathit{opt}}}}
\newcommand{\copt}{\ensuremath{C^{\mathit{opt}}}}
\newcommand{\iopt}{\ensuremath{\mathit{opt}}}

\title{Improved Approximation Algorithms for Matroid and Knapsack Median Problems and
  Applications\footnote{A preliminary version~\cite{Swamy14} appeared in the Proceedings
    of the 17th APPROX, 2014.}} 
\author{
         Chaitanya Swamy\thanks{{\tt cswamy@uwaterloo.ca}.
         Dept. of Combinatorics and Optimization, Univ. Waterloo, Waterloo, ON N2L 3G1.
         Supported in part by NSERC grant 327620-09, an NSERC Discovery Accelerator
         Supplement Award, and an Ontario Early Researcher Award.} 
}
\date{}

\begin{document}

\maketitle

\begin{abstract}
We consider the {\em matroid median} problem~\cite{KrishnaswamyKNSS11}, wherein we are
given a set of facilities with opening costs and a matroid on the facility-set, and
clients with demands and connection costs, and we seek to open an independent set of
facilities and assign clients to open facilities so as to minimize the sum of the
facility-opening  and client-connection costs. We give a simple 8-approximation algorithm
for this problem based on LP-rounding, which improves upon the 16-approximation
in~\cite{KrishnaswamyKNSS11}.  
%
We illustrate the power and versatility of our techniques by deriving: 
(a) an 8-approximation for the {\em two-matroid median} problem, a generalization of
matroid median that we introduce involving two matroids; and 
(b) a 24-approximation algorithm for {\em matroid median with penalties}, which is a vast 
improvement over the 360-approximation obtained in~\cite{KrishnaswamyKNSS11}. 
We show that a variety of seemingly disparate
facility-location problems considered in the literature---data placement problem,
mobile facility location, $k$-median forest, metric uniform minimum-latency
{\footnotesize \textsf{UFL}}---in fact reduce to the matroid median or two-matroid median
problems, and thus obtain {\em improved} approximation guarantees for all these
problems. Our techniques also yield an improvement for the knapsack median problem.
\end{abstract}


\section{Introduction}
We investigate facility location problems wherein the set of open facilities have to
satisfy some matroid independence constraints or knapsack constraints. Specifically, we 
consider the {\em matroid median problem}, which is defined as follows.
As in the uncapacitated facility location problem, we are given a set of facilities $\F$ 
and a set of clients $\D$.
Each facility $i$ has an {\em opening cost} of $f_i$.
Each client $j\in\D$ has demand $d_j$ and assigning client $j$ to facility $i$ incurs an
{\em assignment cost} of $d_jc_{ij}$ proportional to the distance between $i$ and $j$. 
Further, we are given a matroid $M=(\F,\I)$ on the set of facilities. 
The goal is to choose a set $F\in\I$ of facilities to open that forms an independent set
in $M$, and assign each client $j$ to a facility $i(j)\in F$ so as to minimize the total
facility-opening and client-assignment costs, that is, 
$\sum_{i\in F}f_i+\sum_{j\in\D}d_jc_{i(j)j}$.
We assume that the facilities and clients are located in a common metric space, so the
distances $c_{ij}$ form a metric.

The matroid median problem is a generalization of the metric {\em $k$-median} problem,
which is the special case where $M$ is a uniform matroid (and there are no
facility-opening costs), and is thus, \nphard. The matroid median problem without
facility-opening costs was introduced recently by Krishnaswamy et
al.~\cite{KrishnaswamyKNSS11}, who gave a 16-approximation algorithm for this problem.  

Our contributions are threefold.
\begin{list}{$\bullet$}{\usecounter{enumi} \topsep=0.5ex \itemsep=0ex
    \addtolength{\leftmargin}{-2ex}} 
\item We devise an improved 8-approximation algorithm for the matroid-median problem 
(Section~\ref{round}). Moreover, notably, our algorithm is significantly simpler and
cleaner than the one in~\cite{KrishnaswamyKNSS11}, and satisfies the stronger property
that it is a {\em Lagrangian-multiplier-preserving} 8-approximation algorithm (see
Remark~\ref{lmp}). 
The effectiveness and versatility of our simpler approach for matroid median is further
highlighted when we consider some natural extensions of matroid median in
Section~\ref{extn}. We leverage the techniques underlying our simpler and
cleaner algorithm for matroid median to devise: 
(a) an 8-approximation algorithm for the {\em two-matroid median} problem
(Section~\ref{multmat}), which is an extension that we introduce involving two matroids
that captures some interesting facility-location problems considered in the literature; and 
(b) a 24-approximation algorithm (Section~\ref{penalty}) for the {\em matroid median
problem with penalties}, wherein we are allowed to leave client unassigned and incur a
penalty for each unassigned client; this constitutes a vast improvement over the
{approximation ratio of 360 obtained by Krishnaswamy et al.~\cite{KrishnaswamyKNSS11}.}

\item We show that the matroid median and two-matroid median problem turn out to be
rather fundamental problems by showing in Section~\ref{apps} that a variety of facility
location problems that have been considered in the literature can be cast as instances of
matroid median or two-matroid median. These include the data placement
problem~\cite{BaevR01,BaevRS08}, mobile facility
location~\cite{FriggstadS11,AhmadianFS13}, $k$-median forest~\cite{GoertzN11}, and metric 
uniform minimum-latency \ufl~\cite{ChakrabartyS11}. This not only gives a unified
framework for viewing these seemingly disparate problems, but also our approximation
guarantee of 8 {\em yields improved, and in some cases, the first, approximation
guarantees for all these problems}. 

\item We adapt our techniques to also obtain an improvement for the knapsack median  
problem~\cite{KrishnaswamyKNSS11,Kumar12} (Section~\ref{knapmed}). 
\end{list}

Our improvement for matroid median comes from an improved, simpler rounding procedure
for a natural LP relaxation of the problem also considered in~\cite{KrishnaswamyKNSS11}.  
We show that a clustering step introduced
in~\cite{CharikarGTS02} for the $k$-median problem coupled with two applications of the
integrality of the intersection of two submodular (or matroid) polyhedra---one to obtain a
half-integral solution, and another to obtain an integral solution---suffices to
obtain the desired approximation ratio. In contrast, the algorithm
in~\cite{KrishnaswamyKNSS11} starts off with the clustering step in~\cite{CharikarGTS02},
but then further dovetails the rounding procedure of~\cite{CharikarGTS02} creating trees,
then stars, and then applies the integrality of the \nolinebreak
\mbox{intersection of two submodular polyhedra.}

There is great deal of similarity between the 
the rounding algorithm of~\cite{KrishnaswamyKNSS11} for matroid median and
the rounding algorithm of Baev and Rajaraman~\cite{BaevR01} for the data placement
problem, who also perform the initial clustering step in~\cite{CharikarGTS02} and then
create trees and then stars and use these to obtain an integral solution. 
In contrast, our simpler, improved rounding algorithm is similar to the
rounding algorithm in~\cite{BaevRS08} for data placement, who use the initial clustering
step of~\cite{CharikarGTS02} coupled with two min-cost flow computations---one to obtain a 
half-integral solution and another to obtain an integral solution---to obtain the final
solution.  
These similarities are not surprising since, as mentioned above, 
we show in Section~\ref{apps} that the data-placement problem is a special case of the
matroid median problem.  
In fact, our improvements are analogous to those obtained for the data-placement
problem by Baev, Rajaraman, and Swamy~\cite{BaevRS08} over the guarantees
in~\cite{BaevR01}, and stem from similar insights.

A common theme to emerge from our work and~\cite{BaevRS08} is that in various
settings, the initial clustering step introduced by~\cite{CharikarGTS02} imparts
sufficient structure to the fractional solution so that one can then round it using two
applications of suitable integrality-results from combinatorial optimization. 
First, this initial clustering can be used to derive a half-integral solution. 
This was observed explicitly in~\cite{BaevR01} and is implicit
in~\cite{KrishnaswamyKNSS11}, and making this explicit yields significant dividends. 
Second, and this is the oft-overlooked insight (in~\cite{BaevR01,KrishnaswamyKNSS11}), 
a half-integral solution can be easily rounded, and in a better way, {\em without
resorting to creating trees and then stars etc. as in the algorithm
of~\cite{CharikarGTS02}}. This is due to the fact that a half-integral solution is already
``filtered'': if client $j$ is assigned to facility $i$ fractionally, then one can bound
$c_{ij}$ in terms of the assignment cost paid by the fractional solution for $j$ (see
Section~\ref{round}). This enables one to use a standard facility-location clustering step
to set up a suitable combinatorial-optimization problem possessing an integrality
property, and hence, round the half-integral solution. The resulting algorithm is
typically both simpler and has a better approximation ratio than what one would obtain by
{mimicking the steps of~\cite{CharikarGTS02} involving creating trees, stars etc.}

%

Recently, Charikar and Li~\cite{CharikarL12} obtained a 9-approximation algorithm for the
matroid-median problem; our results were obtained independently.%
\footnote{A manuscript containing the 8-approximation for matroid median was
circulated privately in 2012; the current version was posted on the arXiv in Nov. 2013.} 
While there is some similarity between our ideas and those in~\cite{CharikarL12}, we
feel that our algorithm and analysis provides a more illuminating explanation of why
matroid median and some of its extensions (e.g., two-matroid median, matroid median with 
penalties; see Section~\ref{extn}) are ``easy'' to approximate, whereas other variants
such as matroid-intersection median (Section~\ref{extn}) are inapproximable. 
It remains to be seen if our ideas coupled with the dependent-rounding procedure
used in~\cite{CharikarL12} for the $k$-median problem leads to further improvements for
the matroid median problem; we leave this as future work. 

\section{An LP relaxation for matroid median} \label{lp}
We can express the matroid median problem as an integer program and relax the integrality 
constraints to get a linear program (LP). 
Throughout we use $i$ to index facilities in $\F$, and $j$ to index clients in $\D$.
{Let $r$ denote the rank function of the matroid $M=(\F,\I)$.}

\vspace*{-2ex}
\begin{alignat}{3}
\min & \quad & \sum_i f_iy_i &+ \sum_j\sum_i d_jc_{ij}&&x_{ij} \tag{P} 
\label{primal} \\
\text{s.t.} & \quad & \sum_i x_{ij} & \geq 1 && \forall j \label{casgn} \\[-7pt] 
&& \sum_{i\in S} y_i & \leq r(S) && \forall S\sse\F \label{cap} \\
&& 0 \leq x_{ij} & \leq y_i  && \forall i,j. \label{nonneg1} 
\end{alignat}
Variable $y_i$ indicates if facility $i$ is open,
and $x_{ij}$ indicates if client $j$ is assigned to facility $i$. 
The first and third constraints say that each client must be assigned to an open
facility. 
The second constraint encodes the matroid independence constraint.
An integer solution corresponds exactly to a solution to our problem.
We note that \eqref{primal} can be solved in polytime since 
(for example)
a polytime algorithm for submodular-function minimization yields an efficient separation
oracle.

\vspace{-1ex}
\section{A simple 8-approximation algorithm via LP-rounding} \label{round}
\vspace{-1ex}
Let $(x,y)$ denote an optimal solution to \eqref{primal} and $\OPT$ be its value.
We first describe a simple algorithm to round $(x,y)$ to an integer solution losing a
factor of at most 10. In Section~\ref{improved}, we use some additional insights to 
improve the approximation ratio to 8.     
We use the terms connection cost and assignment cost interchangeably. We may assume that
$\sum_i x_{ij}=1$ for every client $j$. 

\vspace{-1ex}
\subsection{Overview of the algorithm} \label{overview}
We first give a high level description of the algorithm. Suppose for a moment that the optimal 
solution $(x,y)$ satisfies the following property: 
\begin{equation}
\text{for every facility $i$, there is {\em at most one} client $j$ such that $x_{ij}>0$.} 
\tag{$*$} \label{prop}
\end{equation}
Let $\F_j=\{i: x_{ij}>0\}$. Notice that the $\F_j$ sets are disjoint. 
We may assume that for $i\in\F_j$, we have $y_i=x_{ij}$, so the objective function is a
linear function of only the $y_i$ variables.
We can then set up the following matroid intersection problem. The first matroid is
$M$ restricted to $\bigcup_j\F_j$. The second matroid $M'$ (on the same ground set
$\bigcup_j \F_j$) is the partition matroid  defined by the $\F_j$ sets; that is, a set is
independent in $M'$ if it contains at most one facility from each $\F_j$. Notice the 
$y_i$-variables yield a fractional point in the {\em intersection of the matroid
polyhedron of $M$ and the matroid-base polyhedron of $M'$}. Since the intersection
of these two polyhedra is known to be integral (see, e.g.,~\cite{CookCPS}), this means
that we can round $(x,y)$ to an integer solution of no greater cost.
Of course, the LP solution need not have property \eqref{prop} so our goal will be to
transform $(x,y)$ to a solution that has this property without increasing the cost by
much.  

Roughly speaking we want to do the following: cluster the clients in $\D$ around certain
`centers' (also clients) such that (a) every client $k$ is assigned to a ``nearby''
cluster center $j$ whose LP assignment cost is less than that of $k$, and (b) the
facilities serving the cluster centers in the fractional solution $(x,y)$ are disjoint. 
So, the modified instance where the demand of a client is moved to the center of its cluster
has a fractional solution, namely the solution induced by $(x,y)$, that satisfies \eqref{prop} 
and has cost at most $\OPT$.  
Furthermore, given a solution to the modified instance we can obtain a solution to
the original instance losing a small additive factor.
One option is to use the decomposition method of Shmoys et al.~\cite{ShmoysTA97} for
uncapacitated facility location (\ufl) that produces precisely such a clustering. The
problem however is that~\cite{ShmoysTA97} uses filtering which involves blowing up the
$x_{ij}$ and $y_i$ values, thus violating the matroid-rank packing constraints.
Chudak and Shmoys~\cite{ChudakS98} use the same clustering idea but without filtering, 
using the dual solution to bound the cost. The difficulty here with this approach 
is that there are terms with negative coefficients in the dual objective
function that correspond to the primal matroid-rank constraints. 
Although~\cite{SwamyS03} showed that it is possible to overcome this difficulty in certain
cases, the situation here looks more complicated and it is not clear how to use their 
techniques.

Instead, we use the clustering technique of Charikar et al.~\cite{CharikarGTS02} to cluster 
clients and first obtain a {\em half-integral solution} $(\hx,\hy)$, that is, every 
$\hx_{ij},\hy_i\in\bigl\{0,\frac{1}{2},1\bigr\}$, to the modified instance with cluster 
centers, losing a factor of 3. Further, any solution here will give a solution to the original 
instance while increasing the cost by at most $4\cdot\OPT$. 
Now we use the clustering method of~\cite{ShmoysTA97} {\em without any filtering}, since 
the half-integral solution $(\hx,\hy)$ is essentially already filtered;
if client $j$ is assigned to $i$ and $i'$ in $\hx$, then
$c_{ij},c_{i'j}\leq 2(c_{ij}\hx_{ij}+c_{i'j}\hx_{i'j})$. 
This final step causes us to lose an additive factor equal to the cost of $(\hx,\hy)$, 
so overall we get an approximation ratio of $4+3+3=10$. In Section~\ref{improved}, we show
that by further exploiting the structure of the half-integral solution, we can give a
better bound on the cost of the integer solution and thus obtain an 8-approximation.

We now describe each of these steps in detail. 
Let $\bC_j=\sum_i c_{ij}x_{ij}$ denote the cost  
incurred by the LP solution to assign one unit of demand of client $j$. Given a vector
$v\in\R^\F$ and a set $S\sse\F$, we use $v(S)$ to denote $\sum_{i\in S}v_i$.

\subsection{Obtaining a half-integral solution \boldmath $(\hx,\hy)$} \label{halfinteg}

\paragraph{Step I: Consolidating demands around centers.} 
We first consolidate (or cluster) the demand of clients at certain clients, that we call
{\em cluster centers}. We do not modify the fractional solution $(x,y)$ but only modify
the demands so that for some clients $k$, the demand $d_k$ is ``moved'' to a ``nearby''
center $j$. We assume every client has non-zero demand 
(we can simply get rid of zero-demand clients).

Set $d'_j\assign 0$ for every $j$. Consider the clients in increasing order of $\bC_j$.
For each client $k$ encountered, if there exists a client $j$ such that $d'_j>0$ and 
$c_{jk}\leq 4\max(\bC_j,\bC_k)=4\bC_k$, set $d'_j\assign d'_j+d_k$, otherwise set $d'_k\assign d_k$. 
Let $D=\{j\in\D: d'_j>0\}$. Each client in $D$ is a cluster center.
Let $\OPT'=\sum_i f_iy_i+\sum_{j\in D,i}d'_jc_{ij}x_{ij}$ denote the cost of $(x,y)$ for the 
modified instance consisting of the cluster centers. 

\begin{lemma} \label{cldm}
(i) If $j,k\in D$, then $c_{jk}\geq 4\max(\bC_j,\bC_k)$,
(ii) $\OPT'\leq\OPT$, and (iii) any solution $(x',y')$ to the modified instance can be converted to a 
solution to the original instance incurring an additional cost of at most $4\cdot\OPT$.
\end{lemma}

\begin{proof}
Suppose $k$ was considered after $j$. Then $d'_j>0$ at this time, otherwise $d'_j$ would remain at 0 
and $j$ would not be in $D$. So if $c_{jk}<4\max(\bC_j,\bC_k)$ then $d'_k$ would remain at 0, giving
a contradiction.
It is clear that if we move the demand of client $k$ to client $j$, then $\bC_j\leq \bC_k$ and 
$c_{jk}\leq 4\bC_k$. So the assignment cost for the new instance, $\sum_j d'_j\bC_j$, only 
decreases and the facility-opening cost $\sum_{i}f_iy_i$ does not change, hence $\OPT'\leq \OPT$.
Given a solution $(x',y')$ to the modified instance, if the demand of $k$ was 
moved to $j$ the extra cost incurred in assigning $k$ to the same facility(ies) as in $x'$ is at 
most $d_kc_{jk}\leq 4d_k\bC_k$ by the triangle inequality, so the total extra cost is at most 
$4\cdot\OPT$.
\end{proof}

From now on we focus on the modified instance with client set $D$ and modified demands
$d'_j$. At the very end we will use the above lemma to translate an integer solution to
the modified instance to an integer solution to the original instance.   

\paragraph{Step II: Transforming to a half-integral solution.}
We define the cluster of a client $j\in D$ to 
be the set $F_j$ of all facilities $i$ 
such that $j$ is the center in $D$ closest to $i$, that is, 
$F_j=\{i: c_{ij}=\min_{k\in D}c_{ik}\}$, with ties broken arbitrarily. 
Let $F'_j\sse F_j=\{i\in F_j: c_{ij}\leq 2\bC_j\}$.
Define $\gm_j:=\min_{i\notin F_j} c_{ij}$, and let $G_j=\{i\in F_j: c_{ij}\leq\gm_j\}$;
see Fig.~\ref{step2}.   
By property (i) of Lemma~\ref{cldm}, we have that $F_j$ contains all the facilities $i$
such that $c_{ij}\leq 2\bC_j$. So $\gm_j\geq 2\bC_j$, $F'_j\sse G_j$, and   
$\sum_{i\in F'_j}x_{ij}=\sum_{i:c_{ij}\leq 2\bC_j}x_{ij}\geq\frac{1}{2}$ by Markov's
inequality.  
Clearly the sets $F_j$ for $j\in D$ are disjoint. 

\begin{figure}[ht!]
\centerline{\resizebox{!}{3in}{\input{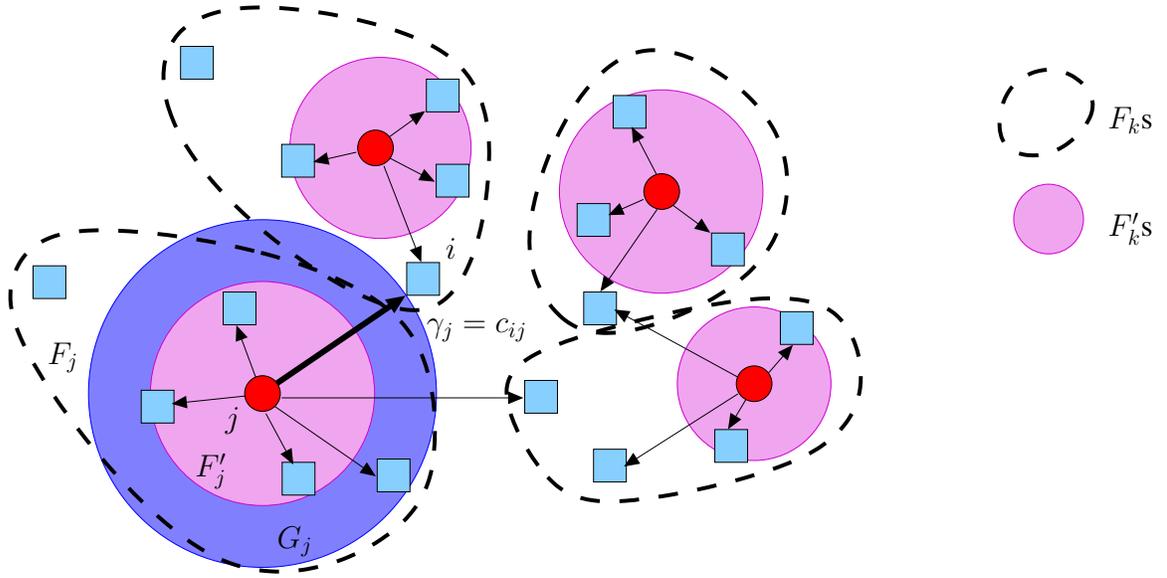}}}
\caption{Illustrating the sets $F_j$, $F'_j$, $G_j$, and the quantity $\gm_j$. Facility
  $i$ is the facility nearest to $j$ notin $F_j$.}
\label{step2}
\end{figure}

To obtain the half-integral solution, we define a suitable vector $y'$ that lies in a
polytope with half-integral extreme points and construct a linear function $T(.)$ such
that $T(y')$ bounds the cost of a fractional solution. We show that 
$T(y')\leq 3\cdot\OPT'$. This implies that one can obtain a ``better'' half-integral
vector $\hy$, which we then argue yields a half-integral solution $(\hx,\hy)$ to the
modified instance of cost at most $T(\hy)\leq T(y')$. 

To motivate the definition of $T(.)$ and the polytope, first define $y'\in\R_+^\F$ as
follows: set $y'_i=x_{i\ell}\leq y_i$ if $i\in G_\ell$, and $y'_i=0$ otherwise.  
Clearly, $y'(F_\ell)=y'(G_\ell)\leq 1$ for every $\ell\in D$. 
Consider some client $j\in D$.
Suppose $\gm_j=c_{ij}$, where $i\in F_k,\ k\neq j$. It is not hard to show that
$c_{i'j}\leq 3\gm_j$ for every facility $i'\in F'_k$ 
(see Lemma~\ref{half}), and so
$\sum_{i\in G_j}c_{ij}y'_i+3\gm_j\bigl(1-y'(G_j)\bigr)\leq 3\bC_j$. 
We use the above linear function of $y'$ as a proxy for $j$'s per-unit-demand assignment 
cost, and define 
$T(v)=\sum_i f_iv_i+\sum_jd'_j\bigl(\sum_{i\in G_j}c_{ij}v_{i}+3\gm_j(1-\sum_{i\in G_j}v_{i})\bigr)$ 
for $v\in\R_+^\F$. 
Note that if $v\in\R_+^\F$ satisfies $v(F'_\ell)\geq 0.5$, $v(G_\ell)\leq 1$ for
all $\ell\in D$, then $d_j\bigl(\sum_{i\in G_j}c_{ij}v_i+3\gm_j(1-v(G_j))\bigr)$ is an
upper bound on $j$'s assignment cost under $v$ (since $1-v(G_j)\leq 0.5\leq v(F'_k)$).
Hence, we define our polytope to be
\vspace{-0.5ex}
\begin{equation}
\Pc:=\Bigl\{v\in\R_+^\F: v(S)\leq r(S) \quad \forall S\sse\F, \qquad 
v(F'_j)\geq\tfrac{1}{2},\ \ v(G_j)\leq 1 \quad \forall j\in D\Bigr\}. \label{halfpoly}
\vspace{-0.5ex}
\end{equation}
\noindent 
We claim that $\Pc$ has half-integral extreme points. The easiest way to see this is to
note that any extreme point of $\Pc$ is defined by a linearly independent system of
tight constraints comprising some $v(S)=r(S)$ equalities corresponding to a laminar set
system, and some $v(F'_j)=\frac{1}{2}$ and $v(G_j)=1$ equalities. The constraint matrix of
this system thus corresponds to equations coming from two laminar set systems; such a
matrix is known to be totally unimodular, and hence the vector $v$ satisfying this
system must be a half-integral solution. (Appendix~\ref{append-halfinteg} gives 
another proof 
{based on the integrality of the intersection of two submodular polyhedra.)}

Clearly $y'\in\Pc$. Hence, we can obtain a half-integral solution
$\hy$ such that $T(\hy)\leq T(y')$.    
For any $j\in D$, 
observe that there is at least one facility $i\in F'_j$ with $\hy_i>0$; we call the
facility $i\in F'_j$ nearest to $j$ the {\em primary facility} of $j$ and set
$\hx_{ij}=\hy_i$. If $\hy_i<1$, then let $i'$ be the facility nearest to $j$ other than
$i$ such that $\hy_{i'}>0$; we call $i'$ the {\em secondary facility} of $j$, and set
$\hx_{i'j}=1-\hx_{ij}$. 
Note that every client in $D$ has a {\em distinct} primary facility. 

\begin{lemma} \label{half}
The cost of $(\hx,\hy)$, that is, $\sum_if_i\hy_i+\sum_{j\in D,i}d'_jc_{ij}\hx_{ij}$,
is at most $3\cdot\OPT'\leq 3\cdot\OPT$.
\end{lemma}

\begin{proof}
We first show that 
$T(y')\leq 3\cdot\OPT'$, and then bound the cost of $(\hx,\hy)$ by
$T(\hy)=\sum_i f_i\hy_i+\sum_jd'_j\bigl(\sum_{i\in G_j}c_{ij}\hy_i+3\gm_j(1-\hy(G_j)\bigr)$.
Since $T(\hy)\leq T(y')$, this proves the lemma.

We have $\OPT'=\sum_i f_iy_i+\sum_j d'_j\bC_j$, and 
for any $j\in D$, we have 
$\bC_j=\sum_{i\in G_j}c_{ij}x_{ij}+\sum_{i\notin G_j}c_{ij}x_{ij} \geq 
\sum_{i\in G_j}c_{ij}x_{ij}+\gm_j(1-\sum_{i\in G_j}x_{ij})$ by the definition of $\gm_j$.
So 
$$
T(y')\leq 
\sum_i f_iy_i+\sum_jd'_j\bigl(\sum_{i\in G_j}c_{ij}x_{ij}+3\gm_j(1-\sum_{i\in G_j}x_{ij})\bigr)
\leq \sum_i f_iy_i+3\sum_j d'_j\bC_j\leq 3\cdot\OPT.
$$

To bound the cost of $(\hx,\hy)$, it suffices to show that the assignment cost of each 
client $j\in D$ is at most $d'_j\bigl(\sum_{i\in G_j}c_{ij}\hy_i+3\gm_j(1-\hy(G_j))\bigr)$.
If $\hy(G_j)=1$, then the assignment cost of $j$ is 
$d'_j\sum_{i\in G_j}c_{ij}\hx_{ij}=d'_j\sum_{i\in G_j}c_{ij}\hy_i$. Otherwise, the assignment
cost of $j$ is at most $d'_j\sum_{i\in G_j}c_{ij}\hx_{ij}+d'_jc_{i'j}(1-\hy(G_j))$, where
$i'$ is the secondary facility of $j$. 
We show that $c_{i'j}\leq 3\gm_j$, which implies the desired bound. 
Let $\gm_j=c_{i''j}$ where $i''\in F_k, k\neq j$. 
Let 
$\ell$ be the primary facility of $k$. Then, $c_{i'j}\leq c_{\ell j}$ and 
$4\max(\bC_j,\bC_k)\leq c_{jk}\leq c_{i''j}+c_{i''k}\leq 2\gm_j$.  
Also $c_{\ell k}\leq 2\bC_k$ since $\ell\in F'_k$. Combining the inequalities we get that 
$c_{i'j}\leq 3\gm_j$.
\end{proof}

%

\vspace{-1.5ex}
\subsection{Converting \boldmath $(\hx,\hy)$ to an integer solution} \label{integer}
Define $\hC_j=\sum_i c_{ij}\hx_{ij}$ and $S_j=\{i: \hx_{ij}>0\}$ for $j\in D$.


\paragraph{Step III: Clustering.}
We cluster the clients in $D$ as follows: pick $j\in D$
with smallest $\hC_j$. Remove every client $k\in D$ such that $S_j\cap S_k\neq\es$; 
we call $j$ the {\em center} of $k$ and denote it by $\ctr(k)$.
Recurse on the remaining set of clients until no client in $D$ is left. Let $D'$ be
the set of clients picked --- these are the {\em new cluster centers}.  
Note that $\ctr(j)=j$ for every $j\in D'$.

\paragraph{Step IV: The matroid intersection problem.}
For convenience, we will say that every client $j\in D$ has both a primary facility,
denoted $i_1(j)$, and a secondary facility, denoted $i_2(j)$, with
$\hx_{i_1(j)j}=\hx_{i_2(j)j}=\frac{1}{2}$,   
with the understanding that if $j$ does not have a secondary facility then 
$i_2(j)=i_1(j)$, and so $\hx_{i_1(j)j}=1$.  
Then we have $\hC_j=\frac{1}{2}(c_{i_1(j)j}+c_{i_2(j)j})$ 
and $c_{i_1(j)j}\leq\hC_j\leq c_{i_2(j)j}\leq 2\hC_j$.


For $i\in\F$, define $\hy'_i=\hx_{ij}\leq\hy_i$ if $i\in S_j$ where $j\in D'$, and
$\hy'_i=\hy_i$ otherwise. 
Then $\hy'$ lies in the polytope
\begin{equation}
\Rc:=\bigl\{z\in\R_+^\F: z(S)\leq r(S) \quad \forall S\sse\F, \qquad
z(S_j)=1 \quad \forall j\in D'\bigr\}. \label{intpoly}
\end{equation}
Observe that $\Rc$ is the intersection of the matroid polytope for $M$ with the matroid
base polytope for the partition matroid defined by the $S_j$ sets for $j\in D'$. This
polytope is known to have integral extreme points. 
Similar to Step II, we define a linear function 
$H(z)=\sum_i f_iz_i+\sum_{k\in D}A_k(z)$, where  
$$
A_k(z)=\begin{cases}
\sum_{i\in S_{\ctr(k)}}d'_kc_{ik}z_i & \text{if $i_1(k)\in S_{\ctr(k)}$} \\
\sum_{i\in S_{\ctr(k)}}d'_kc_{ik}z_i + d'_k\bigl(c_{i_1(k)k}-c_{i_2(k)k}\bigr)z_{i_1(k)} 
& \text{otherwise}. 
\end{cases}
$$
Here, $A_k(z)$ is a proxy for $k$'s assignment cost chosen suitably so that: 
(a) for an integer $\ty\in\Rc$, $A_k(\ty)$ yields an upper bound on $k$'s assignment cost  
(see Lemma~\ref{intbnd}); and
(b) $A_k(\hy')$ is at most $2d'_k\hC_k$ (see Lemma~\ref{rhsbnd}).
Since $\Rc$ is integral, we can find an integer point
$\ty\in\Rc$ such that $H(\ty)\leq H(\hy')$. This yields an integer solution $(\tx,\ty)$ to
the instance with client set $D$, where we assign each client $j\in D'$ to the unique
facility opened from $S_j$, and each client $k\in D\sm D'$ either to $i_1(k)$ if it is
open (i.e., $\ty_{i_1(k)}=1$), or to the facility opened from $S_{\ctr(k)}$.
In Lemma~\ref{intbnd} we prove that the cost of this integer solution is at most
$H(\ty)$, 
and in Lemma~\ref{rhsbnd} we show that $H(\hy')$ 
is at most twice the cost of $(\hx,\hy)$ and hence, at most $6\cdot\OPT$ (by
Lemma~\ref{half}). Combined with Lemma~\ref{cldm}, this yields Theorem~\ref{mmedthm}. 


\begin{lemma} \label{intbnd}
The cost of $(\tx,\ty)$ is at most $H(\ty)\leq H(\hy')$. 
\end{lemma}

\begin{proof}
Clearly, the facility opening cost is $\sum_i f_i\ty_i$, and the assignment cost of a
client $j\in D'$ is $\sum_{i\in S_j}d'_jc_{ij}\ty_i$, which is exactly $A_j(\ty)$. 
Consider a client $k\in D\sm D'$ with $\ctr(k)=j$. 
If $\ty_{i_1(k)}=0$, then the assignment cost of $k$ is $d'_k\sum_{i\in S_j}c_{ik}\ty_i$
which is equal to $A_k(\ty)$. 
If $\ty_{i_1(k)}=1$, then the assignment cost of $k$ is $d'_kc_{i_1(k)k}$.
If $i_1(k)\in S_j$, then $A_k(\ty)=d'_k\sum_{i\in S_j}c_{ik}\ty_i\geq d'_kc_{i_1(k)k}$,
and otherwise 
$A_k(\ty)=d'_k\bigl(\sum_{i\in S_j}c_{ik}\ty_i+c_{i_1(k)k}-c_{i_2(k)k}\bigr)\geq d'_kc_{i_1(k)k}$  
since $i_2(k)$ is the second-nearest facility to $k$, so every facility in $S_j$ is at 
least as far away from $k$ as $i_2(k)$.
\end{proof}

\begin{lemma} \label{rhsbnd}
$H(\hy')$ is at most twice the cost of $(\hx,\hy)$.
\end{lemma}

\begin{proof}
Clearly $\sum_i f_i\hy'_i\leq\sum_i f_i\hy_i$.
For $j\in D'$, we have $A_j(\hy')=\sum_{i\in S_j}d'_jc_{ij}\hx_{ij}$. 
Consider $k\in D\sm D'$ with $\ctr(k)=j$. Let $i'=i_1(j)$ and $i''=i_2(j)$, so
$\hC_j=\frac{1}{2}(c_{i'j}+c_{i''j})\leq\hC_k$. 

If $i_1(k)\in S_j$, then the (at most one) facility $i\in S_j\sm\{i_1(k)\}$
satisfies $c_{ik}\leq c_{i'j}+c_{i''j}+c_{i_1(k)k}\leq 2\hC_k+c_{i''k}$.
So $A_k(\hy')\leq\frac{d'_k}{2}\bigl(2c_{i_1(k)k}+2\hC_k\bigr)\leq 2d'_k\hC_k$.

If $i_1(k)\notin S_j$ then $i_2(k)\in S_j$, so 
$c_{i'k}+c_{i''k}\leq 2c_{i_2(k)k}+c_{i'j}+c_{i''j}\leq 2c_{i_2(k)k}+2\hC_k$.
So $A_k(\hy')$ is at most
$\frac{d'_k}{2}\bigl(2c_{i_2(k)k}+2\hC_k+c_{i_1(k)k}-c_{i_2(k)k}\bigr)=2d'_k\hC_k$.
\end{proof}

\begin{theorem} \label{mmedthm}
The integer solution $(\tx,\ty)$ translates to an integer solution to the original instance
of cost at most $10\cdot\OPT$.
\end{theorem}

\begin{proof}
By Lemmas~\ref{intbnd} and~\ref{rhsbnd}, the cost of $(\tx,\ty)$ (for the modified
instance) is at most twice the cost of $(\hx,\hy)$, and hence, at most $6\cdot\OPT$ by
Lemma~\ref{half}. Applying part (ii) of Lemma~\ref{cldm} yields the theorem.
\end{proof}

\subsection{Improvement to 8-approximation} \label{improved}
The procedure described in Section~\ref{integer} 
shows that {\em any} half-integral solution can be rounded
to an integral one losing a factor of 2 in the cost.
%
We obtain an improved approximation ratio of 8 by exploiting the
structure leading to the half-integral solution obtained in Section~\ref{halfinteg}. 
The key to the improvement comes from the following observation (in various
flavors). 
Consider a non-cluster-center $k\in D'\sm D$ with $\ctr(k)=j$. Let $i$ be a facility
serving both $j$ and $k$. Suppose $i$ is not the primary facility of $k$. 
Without any further information, we can only say that  
$c_{jk}\leq c_{ij}+c_{ik}\leq 3\gm_j+3\gm_k$. However, if we define our half-integral
solution by setting the secondary facility of $k$ to be the primary
facility of the client (in $D$) nearest to $k$, then we have the better bound
$c_{jk}\leq 2\gm_j+2\gm_k$, which yields an improved bound for $k$'s assignment cost.
To push this observation through, we will ``couple'' the rounding steps used to
obtain the half-integral and integral solutions: we tailor the function $T(.)$
(defined in Step II above) so as to allow one to bound the total cost of
the final integral solution obtained. 
Also, we use a different criterion for selecting 
{a cluster center in the clustering performed in Step III.}

The first step is the same as Step I in Section~\ref{halfinteg}. 
Recall that the new client-set is $D$ with demands $\{d'_j\}_{j\in D}$, 
$\OPT'$ is the cost of $(x,y)$ for the modified instance, 
and for each $j\in D$ we define
$F_j=\{i:c_{ij}=\min_{k\in D}c_{ik}\}$, $F'_j=\{i\in F_j:c_{ij}\leq 2\bC_j\}$, 
$\gm_j=\min_{i\notin F_j}c_{ij}$, and $G_j=\{i\in F_j:c_{ij}\leq\gm_j\}$. 

\begin{list}{A\arabic{enumi}.}{\usecounter{enumi} \addtolength{\leftmargin}{-1ex}}
\item {\bf Obtaining a half-integral solution.}
Set $y'_i=x_{ij}\leq y_i$ if $i\in G_j$, and $y'_i=0$ otherwise. 
We define 
\mbox{$T(v)=\sum_i f_iv_i+\sum_jd'_j\bigl(2\sum_{i\in G_j}c_{ij}v_{i}+4\gm_j(1-\sum_{i\in G_j}v_{i})\bigr)$}
for $v\in\R_+^\F$ with some hindsight.
Since $y'$ lies in the half-integral polytope $\Pc$ (see \eqref{halfpoly}), we
can obtain a half-integral $\hy$ such that $T(\hy)\leq T(y')$.

For each client $j\in D$, define $\sg(j)=j$ if $\hy(G_j)=1$, and 
$\sg(j)=\arg\min_{k\in D: k\neq j}c_{jk}$ otherwise (breaking ties arbitrarily).
Note that $c_{j\sg(j)}\leq 2\gm_j$. 
%
As before, we call the facility $i$ nearest to $j$ with $\hy_i>0$ the primary facility of 
$j$ and denote it by $i_1(j)$; we set $\hx_{i_1(j)j}=\hy_{i_1(j)}$. 
Note that $i_1(j)\in F'_j$. 
If $\hy_{i_1(j)}<1$ and $\hy(G_j)=1$, let $i'$ be the fractionally open facility other
than $i_1(j)$ nearest to $j$; otherwise, if $\hy_{i_1(j)}<1$ and $\hy(G_j)<1$, (so
$\sg(j)\neq j$ and $\hy_{i_1(j)}=\frac{1}{2}$), let $i'$ be the primary facility of
$\sg(j)$. We call $i'$ the secondary facility of $j$, and denote it by $i_2(j)$.
Again, for convenience, we consider $j$ as having both a primary and secondary
facility and $\hx_{i_1(j)j}=\hx_{i_2(j)j}=\frac{1}{2}$, with the understanding that if
$\hy_{i_1(j)}=1$, then 
{$i_2(j)=i_1(j)$ and $\hx_{i_1(j)j}=1$.
Let $S_j=\{i:\hx_{ij}>0\}=\{i_1(j),i_2(j)\}$.} 

\medskip
\item {\bf Clustering and rounding to an integral solution.} 
For each $j\in D$, define $C'_j=\bigl(c_{i_1(j)j}+c_{j\sg(j)}+c_{i_2(j)\sg(j)}\bigr)/2$.
We cluster clients as in Step III in Section~\ref{integer}, except that we repeatedly 
pick the client with smallest $C'_j$ among the remaining clients to be the cluster
center. As before, let $D'$ denote the set of cluster centers, and let 
$\ctr(k)=j\in D'$ for $k\in D$ if $k$ was removed in the clustering process because $j$ 
was chosen as a cluster center and $S_j\cap S_k\neq\es$. 

Similar to Step IV in Section~\ref{integer}, for each $i\in\F$, define
$\hy'_i=\hx_{ij}\leq\hy_i$ if $i\in S_j$ where $j\in D'$ and $\hy'_i=\hy_i$ otherwise.  
For $z\in\R_+^\F$, define $H(z)=\sum_i f_iz_i+\sum_{k\in D}L_k(z)$, where 
$$
L_k(z)=\begin{cases}
\sum_{i\in S_{\ctr(k)}}d'_kc_{ik}z_i & \text{if $i_1(k)\in S_{\ctr(k)}$} \\
\sum_{i\in S_{\ctr(k)}}d'_k\bigl(c_{k\sg(k)}+c_{i\sg(k)}\bigr)z_i 
+ d'_k\bigl(c_{i_1(k)k}-c_{k\sg(k)}-c_{i_1(\sg(k))\sg(k)}\bigr)z_{i_1(k)}
& \text{otherwise}.
\end{cases}
$$
As in Step IV in Section~\ref{integer}, $L_k(z)$ is a suitable proxy for $k$'s assignment
cost. It coincides with $A_k(z)$ when $i_1(k)\in S_{\ctr(k)}$; in the other case, we
we have replaced each $c_{ik}$ term for $i\in S_{\ctr(k)}$ in the expression for $A_k(z)$
by the bound $c_{k\sg(k)}+c_{i\sg(k)}$ (note that $i_2(k)=i_1(\sg(k))$). The intent is to
capture the cost savings due to our new definition of $i_2(k)$ but yet ensure that
$L_k(\ty)$ yields an upper bound on $k$'s assignment cost when $\ty\in\Rc$. 
 
Since $\hy'$ lies in the integral polytope $\Rc$ (see \eqref{intpoly}), we can obtain an 
integral vector $\ty$ such that $H(\ty)\leq H(\hy')$, and a corresponding integral
solution $(\tx,\ty)$ (as in Step IV in Section~\ref{integer}).  
\end{list}

\vspace{-3ex}
\paragraph{Analysis.}
By mimicking the proof of Lemma~\ref{half}, we easily obtain that $T(y')\leq 4\cdot\OPT'$. 
Hence, we have $T(\hy)\leq T(y')\leq 4\cdot\OPT'\leq 4\cdot\OPT$. Lemma~\ref{newintbnd}
shows that the cost of $(\tx,\ty)$ is at most $H(\ty)\leq H(\hy')$, and Lemma~\ref{htbnd}
proves that $H(\hy')\leq T(\hy)$. This shows that the cost of $(\tx,\ty)$ is at most
$4\cdot\OPT$. Combined with Lemma~\ref{cldm}, this yields the 8-approximation guarantee
(Theorem~\ref{newthm}).


\begin{lemma} \label{newintbnd}
The cost of $(\tx,\ty)$ is at most $H(\ty)\leq H(\hy')$.
\end{lemma}

\begin{proof}
The facility opening cost is $\sum_i f_i\ty_i$. The assignment cost of a
client $j\in D'$ is $\sum_{i\in S_j}d'_jc_{ij}\ty_i=L_j(\ty)$. 
Consider a client $k\in D\sm D'$ with $\ctr(k)=j$. 
Let $i'=i_1(j),\ i''=i_2(j)$.
If $\ty_{i_1(k)}=0$ or $i_1(k)\in S_j$, then $L_k(\ty)$ is at least 
$d'_k\sum_{i\in S_j}c_{ik}\ty_i$, which is the assignment cost of $k$.
So suppose $\ty_{i_1(k)}=1$ and $i_1(k)\notin S_j$. Then the assignment cost of $k$ is
$d'_kc_{i_1(k)k}$, and since $c_{i\sg(k)}\geq c_{i_1(\sg(k))\sg(k)}$ for every $i\in S_j$, we
have $L_k(\ty)\geq d'_kc_{i_1(k)k}$. 
\end{proof}
  
\begin{lemma} \label{htbnd}
We have $H(\hy')\leq T(\hy)$. 
\end{lemma}

\begin{proofnobox}
Define $B_j(\hy):=d'_j\bigl(2\sum_{i\in G_j}c_{ij}\hy_i+4\gm_j(1-\hy(G_j))\bigr)$.
So $T(\hy)=\sum_i f_i\hy_i+\sum_{j\in D} B_j(\hy)$.
Clearly $\sum_i f_i\hy'_i\leq\sum_i f_i\hy_i$. We show that
$L_j(\hy')\leq B_j(\hy)$ for every $j\in D$, which will complete the proof. 

We first argue that $d'_jC'_j\leq B_j(\hy)$ for every $j\in D$.
If $\hy(G_j)=1$, then $d'_jC'_j=\sum_{i\in G_j}d'_jc_{ij}\hy_i\leq B_j(\hy)$. 
Otherwise, $\hy(G_j)=\frac{1}{2}$, and $c_{j\sg(j)}+c_{i_1(\sg(j))\sg(j)}\leq 3\gm_j$; 
so 
$d'_jC'_j\leq d'_j\bigl(\sum_{i\in G_j}c_{ij}\hy_i+3\gm_j(1-\hy(G_j))\bigr)\leq B_j(\hy)$.

For a client $j\in D'$, we have 
$L_j(\hy')=d'_j\bigl(c_{i_1(j)j}+c_{i_2(j)j}\bigr)/2\leq d'_jC'_j\leq B_j(\hy)$. 
Now consider a client $k\in D\sm D'$. Let $j=\ctr(k)$, and $i'=i_1(j),\ i''=i_2(j)$.  
Note that $C'_j\leq C'_k$.
We consider two cases.
\begin{list}{\arabic{enumi}.}{\usecounter{enumi} \topsep=1ex \itemsep=0.5ex
    \addtolength{\leftmargin}{-2.5ex}}
\item $i_1(k)\in S_j$. This means that $i_1(k)=i''\neq i'$ and $k=\sg(j)$. So 
$$
L_k(\hy')=\frac{d'_k}{2}\cdot\bigl(c_{i''k}+c_{i'k}\bigr)
\leq \frac{d'_k}{2}\cdot\bigl(c_{i'j}+c_{jk}+c_{i''k}\bigr) 
= d'_kC'_j \leq d'_kC'_k \leq B_k(\hy).
$$

\item $i_1(k)\notin S_j$. This implies that $\hy(G_k)=\hy_{i_1(k)}=\frac{1}{2}$.
Let $\ell=\sg(k)$ (which is the same as $j$ if $i_2(k)=i_1(j)$).
We have $L_k(\hy')=\frac{d'_k}{2}\cdot
\bigl(2c_{k\ell}+c_{i'\ell}+c_{i''\ell}+c_{i_1(k)k}-c_{k\ell}-c_{i_1(\ell)\ell}\bigr)$.
If $\ell=j$, then 
$L_k(\hy')=\frac{d'_k}{2}\cdot\bigl(c_{i_1(k)k}+c_{jk}+c_{i''j}\bigr)$.
Notice that $c_{i''j}\leq 2C'_j-c_{i'j}$. 
So we obtain that
$$
L_k(\hy') \leq \frac{d'_k}{2}\cdot\bigl(c_{i_1(k)k}+c_{jk}+2C'_j-c_{i'j}\bigr)
\leq \frac{d'_k}{2}\cdot\bigl(c_{i_1(k)k}+c_{jk}+2C'_k-c_{i'j}\bigr) 
= d'_k\bigl(c_{i_1(k)k}+c_{jk}\bigr).
$$  
If $\ell\neq j$, then $i_2(j)=i''=i_2(k)=i_1(\ell)$, so $\ell=\sg(j)$, and  
$c_{i'j}+c_{j\ell}+c_{i''\ell}=2C'_j\leq 2C'_k=c_{i_1(k)k}+c_{k\ell}+c_{i''\ell}$. 
So $L_k(\hy')\leq\frac{d'_k}{2}\cdot\bigl(c_{i_1(k)k}+c_{k\ell}+c_{j\ell}+c_{i'j}\bigr)
\leq d'_k(c_{i_1(k)k}+c_{k\ell})$.
In both cases, 
\begin{equation*}
L_k(\hy')\leq d'_k(c_{i_1(k)k}+c_{k\sg(k)}\bigr)
\leq d'_k\Bigl(2\sum_{i\in G_k}c_{ik}\hy_i+4\gm_k\bigl(1-\hy(G_k)\bigr)\Bigr) = B_k(\hy). 
\tag*{\qedsymbol}
\end{equation*}
\end{list}
\end{proofnobox}

\vspace{-2ex}
\begin{theorem} \label{newthm}
The integer solution $(\tx,\ty)$ translates to an integer solution to the original instance
of cost at most $8\cdot\OPT$.
\end{theorem}

\begin{remark} \label{lmp}
It is easy to modify the above algorithm to obtain a so-called 
{\em Lagrangian-multiplier preserving} (LMP) 8-approximation algorithm, that is, where the 
solution $(\tx,\ty)$ returned satisfies 
$8\sum_i f_i\ty_i+\sum_{j\in\D,i}d_jc_{ij}\tx_{ij}\leq 8\cdot\OPT$. 
To obtain this, the only change is that we redefine 
$$
T(v)=8\sum_i f_iv_i+\sum_jd'_j\bigl(2\sum_{i\in G_j}c_{ij}v_{i}+4\gm_j(1-\sum_{i\in G_j}v_{i})\bigr), 
\qquad
H(z)=8\sum_i f_iz_i+\sum_{k\in D}L_k(z).
$$
We now have $T(\hy)\leq T(y')\leq 8\sum_if_iy_i+4\sum_{j\in D}d'_j\bC_j$, and 
$8\sum_if_i\ty_i+\sum_{j\in D,i}d'_jc_{ij}\tx_{ij}\leq H(\ty)\leq H(\hy')$. Also, as
before, we have $H(\hy')\leq T(\hy)$. Thus, we have
\begin{eqnarray*}
8\sum_i f_i\ty_i+\sum_{j\in\D,i}d_jc_{ij}\tx_{ij}
& \leq & 8\sum_i f_i\ty_i+\sum_{j\in D,i}d'_jc_{ij}\tx_{ij}+\sum_{j\in\D\sm D}4d_j\bC_j \\
& \leq & 8\sum_if_iy_i+4\sum_{j\in D}d_j\bC_j+8\sum_{j\in\D\sm D}d_j\bC_j
\leq 8\cdot\OPT.
\end{eqnarray*}
\end{remark}

\vspace{-3ex}
\section{Extensions} \label{extn}

\vspace{-1ex}
\subsection{Matroid median with two matroids} \label{multmat}
\vspace{-1ex}
A natural extension of matroid median is the {\em matroid-intersection median} problem,
wherein are given two matroids on the facility-set $\F$, and we require the set of open
facilities to be an independent set in both matroids. 
This problem turns out to be inapproximable to within any multiplicative factor in 
polytime since, as we show in Appendix~\ref{mintmed}, it is \npcomplete to
determine if there is a zero-cost solution; this holds even if one of the matroids is a
partition matroid.   

We consider two extensions of matroid median that are essentially special cases of
matroid-intersection median and can be used to model some interesting problems (see
Section~\ref{apps}). The techniques developed in Section~\ref{round}
readily extend and yield an 8-approximation algorithm (in fact, an LMP 8-approximation)
for both problems.
These extensions may be viewed in some sense as the most-general special cases
of matroid-intersection median that one can hope to \nolinebreak
\mbox{approximately solve in polytime.}
Technically, the key distinction between (general) matroid-intersection median and the
extensions we consider, which enables one to achieve polytime multiplicative 
approximation guarantees for these problems, is the following. In both our extensions, one
can define polytopes analogous to $\Pc$ and $\Rc$ in 
the earlier rounding procedure (see \eqref{halfpoly} and \eqref{intpoly} respectively)
that encode information from the clustering performed in Steps I and III respectively and
whose extreme points are defined by equations coming from {\em two laminar systems}. In
contrast, for matroid-intersection median, the extreme points of the analogous polytopes
are defined by equations coming from three laminar systems (one each from the two
matroids, and one that encodes information about the clustering step), which creates an
insurmountable obstacle. 

The setup in both extensions is similar. We have a matroid $M=(\F,\I)$ on the 
facility-set (and clients with demands and assignment costs). 
$\F$ is partitioned into $\F_1\cup\F_2$ and clients may only be assigned
to facilities in $\F_1$; this can be encoded by setting $c_{ij}=\infty$ for all $i\in\F_2$
and $j\in\D$. We also have lower and upper bounds $(\lbo,\ubo)$, $(\lbt,\ubt)$, and
$(\lb,\ub)$ on the number of facilities that may be opened from $\F_1$, $\F_2$, and $\F$
respectively.
As before, we need to open a feasible set of facilities and assign every client to
an open facility so as to minimize the total facility-opening and client-assignment cost. 
A set $F\sse\F$ of facilities is said to be feasible if: 
(i) $F\in\I$;  
(ii) $\lbo\leq|F\cap\F_1|\leq\ubo$, $\lbt\leq|F\cap\F_2|\leq\ubt$, $\lb\leq|F|\leq\ub$;
and 
(iii) $F\cap\F_2$ satisfies problem-specific constraints.  
While the role of $\F_2$ may seem unclear, notice that a non-trivial 
lower bound on the number of $\F_2$-facilities imposes
restrictions on the facilities that may be opened from $\F_1$ due to the matroid $M$
(see, e.g., $k$-median forest in Section~\ref{apps}).

\vspace{-1ex}
\paragraph{Two-matroid median (\tmmed).} \label{twommed}
In addition to the above setup, we have another matroid $M_2=(\F_2,\I_2)$
on $\F_2$ with rank function $r_2$. 
A set $F$ of facilities is feasible if it satisfies (i) and (ii) above, and
(iii) $F\cap\F_2\in\I_2$. 
We may modify the matroids $M$ and $M_2$ to incorporate the upper bounds $\ub$
and $\ubt$ respectively in their definition; 
{\em we assume that this has been done in the sequel}.   
The LP-relaxation for \tmmed is quite similar to \eqref{primal}. 
We augment \eqref{primal} with the constraints:
\begin{equation*}
y(S)\leq r_2(S) \quad \forall S\sse\F_2, \qquad 
\lbo\leq y(\F_1)\leq\ubo, \quad \lbt\leq y(\F_2), \quad \lb\leq y(\F). 
\end{equation*}
Let $(x,y)$ denote an optimal solution to this LP, and $\OPT$ denote its cost. 
The rounding procedure dovetails the one in Section~\ref{round}. 
The first step is again Step I in Section~\ref{halfinteg}. 
Let $D$ be the new client-set with demands $\{d'_j\}_{j\in D}$, 
$\OPT'$ be the new cost of $(x,y)$, and 
for each $j\in D$, we define $F_j$, $F'_j$, $\gm_j$, and $G_j$ as before.
Note that $F_j\sse\F_1$ for all $j\in D$.

A slight technicality arises in mimicking Step A1 in Section~\ref{improved}: setting
$y'_i=x_{ij}$ for some facility $i\in G_j$ need not 
satisfy the lower-bound constraints. 
To deal with this, for every $j\in D$ and $i\in G_j$ with $0<x_{ij}<y_i$, we replace 
facility $i$ with two co-located ``clones'' $i_1$ and $i_2$. We set $f_{i_1}=f_{i_2}=f_i$, 
$y_{i_1}=x_{ij}=x_{i_1j},\ y_{i_2}=y_i-y_{i_1},\ x_{i_2j}=0$, and for every client 
$k\in D,\ k\neq j$, we arbitrarily split $x_{ik}$ into $x_{i_1k}\leq y_{i_1}$ and
$x_{i_2k}\leq y_{i_2}$ so that $x_{i_1k}+x_{i_2k}=x_{ik}$. We define a new set $G'_j$
consisting of the new facilities $i$ for which $c_{ij}=\min_{k\in D}c_{ik}$,
$c_{ij}\leq\gm_j$ {\em and $x_{ij}=y_i>0$} (that is, $G'_j$ consists of the new
$i_1$-clones and the old facilities $i\in G_j$ with $x_{ij}=y_i>0$).
We continue to let $F'_j$ denote the facilities $i$ with $c_{ij}=\min_{k\in D}c_{ik}$
and $c_{ij}\leq 2\bC_j$.
Let $\F'_1$ denote the new $\F_1$-set after these changes, and $\F'=\F'_1\cup\F_2$; 
the bounds $\lbo,\ \ubo,\ \lb,\ \ub$ are unchanged.  
Set $h(i)=\{i_1,i_2\}$ if $i$ is cloned into $i_1,\ i_2$, and $h(i)=\{i\}$ otherwise.
We update the rank function $r$ to $r'$ (over $2^{\F'}$) in the obvious way:  
$r'(S)=r\bigl(\{i\in\F: h(i)\cap S\neq\es\}\bigr)$. Note that $r'$ defines a matroid on
$\F'$. Clearly, a solution to the modified translates to a solution to the original
instance and vice versa.

We continue with steps A1, A2 in Section~\ref{improved}, replacing $G_j$ with $G'_j$, and 
using suitable polytopes in place of $\Pc$ and $\Rc$ to obtain the half-integral and
integral solutions. To obtain a half-integral solution, we define
\begin{eqnarray}
\Pc':=\biggl\{\ & v\in\R_+^{\F'} : \ v(S)\leq r'(S) \ \ \forall S\sse\F', \quad 
v(S)\leq r_2(S) \ \ \forall S\sse\F_2, \quad \lb\leq v(\F') \notag \\[-1ex]
& \lbo\leq v(\F'_1)\leq\ubo, \quad \lbt\leq v(\F_2), 
\quad v(F'_j)\geq\tfrac{1}{2}, \ v(G'_j)\leq 1 \quad \forall j\in D\ \biggr\}. 
\label{newhalfpoly}
\end{eqnarray}
%
\noindent Clearly (the new vector) $y$ lies in $\Pc'$. 
The key observation is that an extreme point of $\Pc'$ is again defined by a linearly
independent system of tight constraints coming from two laminar systems: one consisting of
some tight $v(S)\leq r'(S)$ and 
$\lb\leq v(\F')\leq\ub$ constraints; the other consisting of some tight $v(S)\leq r_2(S)$
and $\lbo\leq v(\F'_1)\leq\ubo,\ \lbt\leq v(\F_2)\leq\ubt$ constraints,
and some tight $v(F'_j)\leq\frac{1}{2}$ and $v(G'_j)\geq 1$ constraints. 
Thus, $\Pc'$ has half-integral extreme points, and so we can
find a half-integral $\hy$ such that $T(\hy)\leq T(y)$ (where $T(.)$ is as defined in
Section~\ref{improved}), and a corresponding solution $(\hx,\hy)$ as in step A1. 

We round $(\hx,\hy)$ to an integral solution as in step A2. 
Recall that $S_j=\{i:\hx_{ij}>0\}$. 
We define $C'_j$ and cluster clients in $D$ as in step A2 (again using $G'_j$ instead of 
$G_j$) to obtain the set $D'$ of cluster centers.
A useful observation is that if $|S_j|=1$ then we may assume that $j\in D'$. This
is because for any $k\in D$ with $S_k\cap S_j\neq\es$, we have $\sg(k)=j$ and
therefore $C'_k\geq\bigl(c_{jk}+c_{i_1(j)j}\bigr)/2\geq c_{i_1(j)j}=C'_j$.
Thus, if $j\in D'$, then $\hx_{ij}=\hy_i$ for all $i\in S_j$: this is clearly true if
$|S_j|=1$; otherwise, we have that $|S_{\sg(j)}|=2$ (since $\sg(j)\notin D'$) and
so $\hy_{i_1(\sg(j))}=\frac{1}{2}$. The polytope used to round $\hy$ is
\begin{equation}
\begin{split}
\Rc':=\biggl\{\ & z\in\R_+^{\F'} : \ z(S)\leq r'(S) \quad \forall S\sse\F', \qquad
z(S)\leq r_2(S) \quad \forall S\sse\F_2, \\[-1ex]
& \lbo\leq z(\F'_1)\leq\ubo, \quad \lbt\leq z(\F_2),  
\quad \lb\leq z(\F'), \quad z(S_j)=1 \ \ \forall j\in D'\ \biggr\} 
\end{split}\label{newintpoly}
\end{equation}
%
\noindent which has integral extreme points. 
So we obtain an integral vector $\ty$ such that $H(\ty)\leq H(\hy)$
(were $H(.)$ is as defined in in Section~\ref{improved}), and hence an
integer solution $(\tx,\ty)$. 
Mimicking the analysis in Section~\ref{improved}, 
we obtain that $T(\hy)\leq T(y)\leq 4\cdot\OPT'$, and the cost of
$(\tx,\ty)$ is at most $H(\ty)\leq H(\hy)\leq T(\hy)$. 
Thus, we obtain the following theorem.  

\begin{theorem} \label{tmedthm}
The integer solution $(\tx,\ty)$ yields an integer solution to \tmmed of cost at most
$8\cdot\OPT$. 
\end{theorem}

\paragraph{Laminarity-constrained matroid median (\lmmed).} \label{laminar}
In \lmmed, in addition to the common setup,
we have a laminar family $\Lc$ on $\F_2$ and bounds $0\leq\ell_S\leq u_S$ for
every set $S\in\Lc$; a set $F$ of facilities is feasible if it satisfies (i) and (ii)
above, and 
(iii) $\ell_S\leq |F\cap S|\leq u_S$ for all $S\in\Lc$,  

The approach used for \tmmed 
also works for \lmmed. The only (obvious) changes are that the LP-relaxation, as well as
the definition of the polytopes $\Pc'$ and $\Rc'$ (in \eqref{newhalfpoly} and
\eqref{newintpoly}) now include the laminarity constraints in place of the rank
constraints for the second matroid. All other steps and arguments proceed 
{\em identically}, and so we obtain an 8-approximation algorithm for
laminarity-constrained matroid median.

\subsection{Matroid median with penalties} \label{penalty}
This is the generalization of matroid median where are allowed to leave some clients
unassigned at the expense of incurring a penalty $d_j\pi_j$ for each unassigned client $j$. 
This changes the LP-relaxation \eqref{primal} as follows. We use a variable $z_j$ for each
client $j\in\D$ to denote if we incur the penalty for client $j$, and modify the
assignment constraint for client $j$ to $\sum_i x_{ij}+z_j\geq 1$; also the objective is  
now to minimize $\sum_i f_iy_i+\sum_{j}d_j\bigl(\sum_i c_{ij}x_{ij}+\pi_jz_j\bigr)$. 
Let $(x,y,z)$ denote an optimal solution to this LP and $\OPT$ be its value.

%
Krishnaswamy et al.~\cite{KrishnaswamyKNSS11} showed that $(x,y,z)$ can be rounded to an
integer solution losing a factor of 360. 
We show that our rounding approach for matroid median can be adapted to yield a
substantially improved 24-approximation algorithm.  
The rounding procedure is similar to the one described in Section~\ref{round} for 
matroid median, except that we now need to deal with the complication that a
client need be assigned fractionally to an extent of 1. 

Let $X_j=\sum_i x_{ij}$, $\bC_j=\sum_i c_{ij}x_{ij}/X_j$, and 
$\lp_j=\sum_i c_{ij}x_{ij}+\pi_jz_j=\bC_jX_j+\pi_jz_j$. 
We may assume that $X_j+z_j=1$ for every client $j$ and that if $x_{ij}>0$ then
$c_{ij}\leq\pi_j$, so we have $\bC_j\leq\lp_j\leq\pi_j$.


\paragraph{Step 0.} First, we set $\tz_j=1$ and incur the penalty for each client $j$ for 
which 
$\pi_j\leq 2\lp_j$. 
In the sequel, we work with the remaining set $\D'=\{j\in\D: 2\lp_j<\pi_j\}$ of
clients. Note that $X_j>\frac{1}{2}$ for every $j\in\D'$.
Let $\OPT''=\sum_i f_iy_i+\sum_{j\in\D'}d_j\bigl(\sum_i c_{ij}x_{ij}+\pi_jz_j\bigr)$.

\paragraph{Step I: Consolidating demands.} We consolidate demands around centers in a
manner similar to Step I of the rounding procedure in Section~\ref{round}. The difference
is that if $k$ is consolidated with client $j$, 
then we cannot simply add $d_k$ to $j$'s demand and replicate $j$'s
assignment for $k$ (since $\pi_k$ could be much larger than $\pi_j$ so that
$\bC_jX_j+\pi_k(1-X_j)$ need not be bounded in terms of $\lp_k$). Instead, we treat $k$ as 
being co-located with $j$ and recompute $k$'s assignment. 

Let $L$ be a list of clients in $\D'$ arranged in
increasing order of $\lp_j$. Let $D=\es$.
We compute a new assignment $(x',z')$ for the clients as follows. Set $x'_{ij}=z'_j=0$ for
all $i,j$. Remove the first client $j\in L$ and add it to $D$. 
Set $x'_{ij}=x_{ij}$ for all facilities $i$ and $z'_j=z_j$; also set $\nbr(j)=j$. 
For every client $k$ in $L$ with $c_{jk}\leq 4\lp_k$, we 
remove $k$ from $L$, and set $\nbr(k)=j$. We consider $k$ to be co-located with $j$ and
re-optimize $k$'s assignment. So we set $x'_{ik}=y_i$ starting from the facility nearest to
$j$ and continuing until $k$ is completely assigned or until the last facility $i$ such
that $c_{ij}\leq\pi_k$, in which case we set $z'_k=1-\sum_i x'_{ik}$. 
Note that 
$\sum_i c_{ij}x'_{ik}+\pi_kz'_k\leq\sum_ic_{ij}x_{ik}+\pi_kz_k\leq 4\lp_k+\lp_k$. 

We call each client in $D$ a {\em cluster center}.
Let $\{c'_{ij}\}$ denote the assignment costs of the clients with respect to their new
locations. Let 
$\OPT'=\sum_if_iy_i+\sum_{j\in\D'}d_j\bigl(\sum_ic'_{ij}x'_{ij}+\pi_jz'_j\bigr)$ 
denote the cost of the modified solution for the modified instance.
The following lemma is immediate.

\begin{lemma} \label{penalty-cldm}
The following hold: (i) if $j,k\in \D'$ are not co-located, then 
$c_{jk}\geq 4\max(\lp_j,\lp_k)$,
(ii) $\OPT'\leq 5\cdot\OPT''$, and 
(iii) any solution to the modified instance can be converted to a solution to the original
instance involving client-set $\D'$ incurring an additional cost of at most
$4\cdot\OPT''$.   
\end{lemma}

\paragraph{Step II: Obtaining a half-integral solution.}
As in Step II of Section~\ref{round}, we define a suitable vector $y'$ that lies in a 
polytope with half-integral extreme points and construct a linear function $T(.)$ with
$T(y')=O(\OPT')$ bounding the cost of a fractional solution. 
We can then obtain a ``better'' half-integral vector $\hy$, which yields a half-integral
solution. 
In Step III, we round $\hy$ to an integral solution whose cost we argue is at most 
$T(\hy)\leq T(y')$.  

Consider a client $j\in D$. 
Let $F_j=\{i: c_{ij}=\min_{k\in D}c_{ik}\}$, 
$F'_j=\{i\in F_j: c_{ij}\leq 2\lp_j\}$,
$\gm_j=\min_{i\notin F_j} c_{ij}$, and $G_j=\{i\in F_j: c_{ij}\leq\gm_j\}$. 
Note that $\sum_{i\in F'_j}x'_{ij}=\sum_{i:c_{ij}\leq 2\lp_j}x'_{ij}>\frac{1}{2}$ since
$\lp_j\geq\sum_{i:c_{ij}>2\lp_j}c_{ij}x'_{ij}+\pi_jz'_j$ (and $\pi_j>2\lp_j$) implies that 
$\sum_{i:c_{ij}>2\lp_j}x'_{ij}+z'_j<\frac{1}{2}$.  
Consider the facilities in $G_j$ in increasing order of their distance from $j$. 
For every facility $i\in G_j$, we set 
$y'_i=\min\{y_i,1-\sum_{i'\in G_j: i'\text{ comes before }i}y'_{i'}\}$.
We set $y'_i=0$ for all other $i\in F_j$.
Note that $y'(F'_j)=\min\bigl\{1,y(F'_j)\bigr\}\geq\frac{1}{2}$.
Clearly, $y'(F_j)=y'(G_j)\leq 1$ and 
if $y(G_j)\leq 1$, then $y'_i=y_i$ for all $i\in G_j$.

Given $v\in\R_+^\F$, for a client $k\in\D'$ with $\nbr(k)=j$, we define  
$B_k(v)=d_k\bigl(\sum_{i\in G_j:c'_{ik}\leq\pi_k}2c'_{ik}v_i
+\min\{2\pi_k,4\gm_j\}(1-\sum_{i\in G_j:c'_{ik}\leq\pi_k}v_i)\bigr)$.
Now set $T(v)=\sum_i f_iv_i+\sum_{j\in\D'}B_j(v)$. Clearly $y'\leq y$, so
$y'$ lies in the polytope $\Pc$ (see \eqref{halfpoly}), 
which has half-integral extreme points. 
So we can obtain a half-integral point $\hy\in\Pc'$ such that $T(\hy)\leq T(y')$.

We now obtain a half-integral assignment for the clients in $\D'$ as follows.
Consider a client $k$ and let $j=\nbr(k)$. (Note that we could have $k=j$.)
Set $\sg(j)$ to be $j$ if $\hy(G_j)=1$, and $\arg\min_{\ell\in D: \ell\neq j}c_{j\ell}$
otherwise (as in Section~\ref{improved}).
Call the facility $i\in F'_j$ nearest to $j$ the primary facility of $k$, and set
$\hx_{ik}=\hy_i$. 
If $\hy_i<1$, then define $i'$ to be the facility nearest to $j$ other than $i$ with
$\hy_{i'}>0$ if $\hy(G_j)=1$, and the primary facility of $\sg(j)$ otherwise.
If $\hy\bigl(\{i''\in G_j: c_{i''j}\leq\pi_k\}\bigr)=\frac{1}{2}$ and $\pi_k\leq 2\gm_j$,
we set $\hz_k=\frac{1}{2}=1-\hx_{ik}$.
Otherwise, 
we set $\hx_{i'k}=\frac{1}{2}=1-\hx_{ik}$ and call $i'$ the secondary facility of $k$.  

\paragraph{Step III: Rounding $(\hx,\hy)$ to an integer solution.}
This step is quite straightforward. 
We incur the penalty for all clients $j\in\D'$ with $\hz_j=\frac{1}{2}$. 
Note that all the remaining clients $k$ with $\nbr(k)=j$ are (co-located and) assigned
identically and completely in $(\hx,\hy,\hz)$. Viewing this as an instance with demand
consolidated at the cluster centers, we use the rounding procedure in 
step A2 of Section~\ref{improved} to convert the half-integral solution of these remaining 
clients into an integral one. 
Let $(\tx,\ty,\tz)$ denote the resulting integer solution.


\begin{lemma} \label{penalty-half}
We have 
$T(\hy)\leq T(y')\leq 4\cdot\OPT'$.
\end{lemma}

\begin{proof} 
It suffices to show that for every client $k$, we have 
$B_k(y')\leq 4d_k\bigl(\sum_i c'_{ik}x'_{ik}+\pi_kz'_k\bigr)$. Let $j=\nbr(k)$.
Consider the facilities in $G_j$ in increasing order of their distance from $j$.
If $\pi_k<\gm_j$, then (we may assume that) $k$ uses the facilities in $G_j$ with
$c'_{ik}=c_{ij}\leq\pi_k$ fully (i.e., $x'_{ik}=y_i$) until either it is completely
assigned (and the last facility used by $k$ may be partially used) or we exhaust the
facilities in $G_j$ with $c_{ij}\leq\pi_k$. 
In both cases, we have $x'_{ik}=y'_i$ for all $i\in G_j$ with $c'_{ik}\leq\pi_k$ and
$z'_k=1-\sum_{i\in G_j:c'_{ik}\leq\pi_k}y'_i$, and so 
$\sum_i c'_{ik}x'_{ik}+\pi_kz'_k=
\sum_{i\in G_j:c'_{ik}\leq\pi_k}c'_{ik}y'_i+\pi_k\bigl(1-\sum_{i\in G_j:c'_{ik}\leq\pi_k}y'_i\bigr)$. 
If $\pi_k\geq\gm_j$, then 
$\sum_i c'_{ik}x'_{ik}+\pi_kz'_k\geq
\sum_{i\in G_j}c'_{ik}x'_{ik}+\gm_j\bigl(1-\sum_{i\in G_j}x'_{ik}\bigr)$;
also, $x'_{ik}=y'_i$ for all $i\in G_j$ since $\pi_k\geq c'_{ik}$.  
So in every case, we have $B_k(y')\leq 4d_k\bigl(\sum_i c'_{ik}x'_{ik}+\pi_kz'_k\bigr)$.
\end{proof}

\begin{lemma} \label{penalty-intbnd}
The cost of $(\tx,\ty,\tz)$ for the modified instance is at most 
$T(\hy)$. 
\end{lemma}

\begin{proof} 
Consider a client $k$ and let $j=\nbr(k)$.
If $\tz_k=1$ then $B_k(\hy)\geq\pi_k$, since $\hz_k=\frac{1}{2}$
implies that $\hy(N_k)=\frac{1}{2}$, where $N_k=\{i\in G_j: c_{ij}\leq\pi_k\}$, and
$\pi_k<2\gm_j$. 
If $\tz_k=0$ then we claim that
$B_k(\hy)=d_k\bigl(2\sum_{i\in G_j}c_{ij}\hy_i+4\gm_j(1-\hy(G_j))\bigr)$.
If $\hy(N_k)=\frac{1}{2}$, then this follows since we must have $\pi_k>2\gm_j$ for $\hz_k$
to be 0; otherwise, $\hy(N_k)=1=\hy(G_j)$ and again the equality holds. 

The proof of Lemma~\ref{newintbnd} 
now shows that 
$\sum_i f_i\ty_i+\sum_{k:\hz_k=0}c'_{ik}\tx_{ik}\leq 
H(\hy)\leq\sum_if_i\hy_i+\sum_{k:\hz_k=0}B_k(\hy)$, where $H(.)$ is the function defined
in step A2 of Section~\ref{improved} for the instance where each cluster center $j$ has 
demand $d'_j:=\sum_{\substack{k:\nbr(k)=j \\ \hz_k=0}}d_k$. 
Hence, the total cost of $(\tx,\ty,\tz)$ for the modified instance is at most $T(\hy)$.
\end{proof}

Combined with parts (ii) and (iii) of Lemma~\ref{penalty-cldm}, we obtain a solution to
the original instance involving client-set $\D'$ of cost at most $24\cdot\OPT''$. 
Adding in the penalties of the clients in $\D\sm\D'$ (recall that $\pi_j\leq 2\lp_j$
for each $j\in\D\sm\D'$), we obtain that the total cost is at most $24\cdot\OPT$.

\begin{theorem} \label{penalty-thm}
One can round $(x,y,z)$ to an integer solution of cost at most $24\cdot\OPT$.
\end{theorem}

\section{Applications} \label{apps}
We now show that the various facility location problems listed below can be cast as
special cases of matroid median or the extensions considered in Section~\ref{multmat}.
Thus, our 8-approximation algorithms for matroid median and these extensions immediately
yield {\em improved approximation guarantees for all these problems}.     

\setlength{\tabcolsep}{1ex}
\medskip 
\begin{tabular}{|p{2in}|p{3.3in}|} \hline
{\bf Problem} & {\bf Previous best approximation factor} \\ \hline
Data placement problem~\cite{BaevR01,BaevRS08} & 10~\cite{BaevRS08} \\ \hline
{Mobile facility location~\cite{FriggstadS11,AhmadianFS13} (with general
  movement costs)} &  
{---; {\em our reduction} and results
of~\cite{KrishnaswamyKNSS11,CharikarL12} yield factors of 16 and 9 
\quad ($(3+\e)$~\cite{AhmadianFS13} for proportional movement costs)} \\ \hline
$k$-median forest~\cite{GoertzN11} (with non-uniform metrics) &
\multirow{2}{*}{16~\cite{GoertzN11} \quad ($(3+\e)$~\cite{GoertzN11} for related metrics)}
\\ \hline 
Metric-uniform minimum-latency \ufl (\mlufl)~\cite{ChakrabartyS11} &
\multirow{2}{*}{10.773~\cite{ChakrabartyS11}} \\ \hline
\end{tabular}


\paragraph{The data placement problem.} 
We have a set of caches $\F$, a set of data objects $\Oc$, and a set of clients $\D$.   
Each cache $i\in\F$ has a capacity $u_i$. Each client $j\in\D$ has demand $d_j$ for a
specific data object $o(j)\in\Oc$ and has to be assigned to a cache that stores $o(j)$. 
Storing an object $o$ in cache $i$ incurs a storage cost of $f_i^o$, and
assigning client $j$ to cache $i$ incurs an access cost of $d_jc_{ij}$, where the
$c_{ij}$s form a metric.   
We want to determine a set of objects $\Oc(i)\subseteq\Oc$ to place in each cache $i\in\F$
satisfying $|\Oc(i)|\leq u_i$, and assign each client $j$ to a cache $i(j)$ that stores
object $o(j)$, (i.e., $o(j)\in \Oc(i(j))$) 
{so as to minimize $\sum_{i\in\F}\sum_{o\in \Oc(i)}f_i^o+\sum_{j\in\D}d_jc_{i(j)j}$.} 

\medskip
\noindent
{\it Reduction to matroid median.\ }
The facility-set in the matroid-median instance is
$\F\times\Oc$. Facility $(i,o)$ denotes that we store object $o$ in cache $i$, and has
cost $f_i^o$. The client set is $\D$. We set the distance $c_{(i,o)j}$ 
to be $c_{ij}$ if $o(j)=o$ and $\infty$ otherwise, thus enforcing that each client $j$ is
only assigned to a facility containing object $o(j)$. 
The new distances form a metric if the $c_{ij}$s form a metric. 
The cache-capacity constraints are incorporated via the matroid 
where a set $S\sse\F\times\Oc$ is independent if $|\{(i',o)\in S: i'=i\}|\leq u_i$ for
every $i\in\F$. 

\paragraph{Mobile facility location.} 
In the version with general movement costs, the input is a 
metric space $\bigl(V,\{c_{ij}\}\bigr)$. We have a set $\D\sse V$ of clients, with each
client $j$ having demand $d_j$, and a set $\F\sse V$ of initial facility locations. A
solution moves each facility $i\in\F$ to a final location $s_i\in V$ incurring a movement
cost of $w_{is_i}\geq 0$, and assigns each client $j$ to the final location $s$ of some
facility incurring an assignment cost of $d_jc_{sj}$. The goal is to minimize the sum of
all the movement and assignment costs. 
%

\medskip
\noindent
{\it Reduction to matroid median.\ } We define the
facility-set in the matroid-median instance to be $\F\times V$. Facility $(i,s_i)$ denotes
that $i\in\F$ is moved to location $s\in V$, and has cost $w_{is}$ (note that $s$
could be $i$).%
The client-set is unchanged, and we set $c_{(i,s)j}$ to be $c_{sj}$ for
every facility $(i,s)\in\F\times V$ and client $j\in\D$. These new distances form a
metric: we have $c_{(i,s)j}\leq c_{(i,s)k}+c_{(i',s')k}+c_{(i',s')j}$
since $c_{sj}\leq c_{sk}+c_{s'k}+c_{s'j}$. The constraint that a facility
in $\F$ can only be moved to one final location can be encoded by defining a matroid 
where a set $S\sse\F\times V$ is said to be independent if $|\{(i',s)\in S: i'=i\}|\leq 1$
for all $i\in\F$.%
\footnote{We are assuming here that $w_{ii}=0$ for every $i\in\F$, so that not opening any
facility in $\{i\}\times V$ correctly encodes that $i$ is not moved and no clients are
assigned to it in the mobile-facility-location instance.
This condition is without loss of generality. If $w_{ii}\neq 0$ then if $r\in V$
is such that $w_{ir}=\min_{s\in V} w_{is}$, we can ``move'' $i$ to $r$ (i.e., set
$\F\assign\F\sm\{i\}\cup\{r\}$ making a copy of $r$ if $r$ was previously in $\F$), and
define the movement cost of $r$ to be $w'_{rs}=w_{is}-w_{ir}$ for all $s\in V$. It is easy
to see that a $\rho$-approximate solution to the new instance translates to a
$\rho$-approximate solution to the original instance.}

\paragraph{$k$-median forest.} In the non-uniform 
version, we have two metric spaces $\bigl(V,\{c_{uv}\}\bigr)$ and
$\bigl(V,\{d_{uv}\}\bigr)$. 
The goal is to find $S\sse V$ with $|S|\leq k$ and assign every node
$j\in V$ to $i(j)\in S$ so as to minimize 
$\sum_j c_{i(j)j}+d\bigl(\mmst(V/S)\bigr)$, where $\mmst(V/S)$ is a
{minimum spanning 
forest where each component contains a node of $S$.} 

\medskip \noindent
{\it Reduction to \tmmed (or \lmmed).\ }
We actually reduce a generalization, where there is an ``opening cost'' $f_i\geq 0$
incurred for including $i$ in $S$;
the resulting instance is also an \lmmed instance. 
We add a root $r$ to $V$. The facility-set $\F$ is the
edge-set of the complete graph on $V\cup\{r\}$. The client-set is $\D:=V$.
Selecting a facility $(r,i)$ denotes that $i\in S$, and selecting a facility $(u,v)$,
where $u,v\neq r$, denotes that $(u,v)$ is part of $\mmst(V/S)$. We let $\F_1$ be the 
edges incident to $r$, and $\F_2$ be the remaining edges. The cost of a facility
$(r,i)\in\F_1$ is $f_i$; the cost of a facility $(u,v)\in\F_2$ is $d_{uv}$. The
client-facility distances are given by $c_{(r,i)j}=c_{ij}$ and $c_{ej}=\infty$ for every
$e\in\F_2$. Note that these $\{c_{ej}\}$ distances form a metric. We let $M$ be the
graphic matroid of the complete graph on $V\cup\{r\}$. We impose a lower bound of
$|V|$ on the number of facilities opened from $\F$, and an upper bound of $k$ on the
number of facilities opened from $\F_1$. The matroid $M_2$ on $\F_2$ is the vacuous one
where every set is independent.

A feasible solution to the \tmmed instance corresponds to a spanning tree on $V\cup\{r\}$
where $r$ has degree at most $k$. This yields a solution to $k$-median forest of
no-greater cost, where the set $S$ is the set of nodes adjacent to $r$ in this edge-set. 
Conversely, it is easy to see that a solution $S$ to the $k$-median forest instance yields
a \tmmed solution of no-greater cost.

\paragraph{Metric uniform \mlufl.}
We have a set $\F$ of facilities with opening costs $\{f_i\}_{i\in\F}$, and a
set $\D$ of clients with assignment costs $\{c_{ij}\}_{j\in\D,i\in\F}$, where the
$c_{ij}$s form a metric. Also, we have a monotone latency-cost function
$\ld:\Z_+\mapsto\R_+$. The goal is to choose a set $F\sse\F$ of facilities to open,
assign each open facility $i\in F$ a distinct time-index $t_i\in\{1,\ldots,|\F|\}$, and
assign each client $j$ to an open facility $i(j)\in F$ so as to minimize 
$\sum_{i\in F}f_i+\sum_{j\in\D}\bigl(c_{i(j)j}+\ld(t_{i(j)})\bigr)$. 

\medskip \noindent
{\it Reduction to matroid median.\ }
We define the facility-set to be
$\F\times\{1,\ldots,|\F|\}$ and the matroid on this set to encode that a set $S$ is
independent if $|\{(i,t')\in S: t'=t\}|\leq 1$ for all $t\in\{1,\ldots,|\F|\}$.
We set $f_{(i,t)}=f_i$ and $c_{(i,t),j}=c_{ij}+\ld(t)$; note that these distances
form a metric. It is easy to see that we can convert any matroid-median solution to one
where we open at most one $(i,t)$ facility for any given $i$ without increasing the cost,
and hence, the matroid-median instance correctly encodes metric uniform \mlufl.

\section{Knapsack median} \label{knapmed}
We now consider the {\em knapsack median problem}~\cite{KrishnaswamyKNSS11,Kumar12},
wherein instead of a matroid on the facility-set, we have a knapsack constraint on the
facility-set. 
Kumar~\cite{Kumar12} obtained the first constant-factor approximation algorithm for this
problem, and~\cite{CharikarL12} obtained an improved 34-approximation algorithm. 
We consider a somewhat more-general version of knapsack median, 
wherein each facility $i$ has a facility-opening cost $f_i$ and a {\em weight} $w_i$, and
we have a knapsack constraint $\sum_{i\in F}w_i\leq B$ constraining the total weight of
open facilities. We leverage the ideas from our simpler improved rounding
procedure for matroid median to obtain an improved 32-approximation algorithm for
this (generalized) knapsack-median problem. 
We show that one can obtain a nearly half-integral solution whose cost is within a
constant-factor of the optimum. It then turns out to be easy to round this to an integral
solution.  
The resulting algorithm and analysis is simpler than that in~\cite{Kumar12,CharikarL12}.
We defer the details to Appendix~\ref{append-knap}.    

\section*{Acknowledgments}
I thank Deeparnab Chakrabarty for various stimulating discussions that eventually led to
this work. I thank Chandra Chekuri for some useful discussions regarding the
matroid-intersection median problem. 

\appendix

\section{Alternate proof of half-integrality of the polytope \boldmath $\Pc$ defined by
  \eqref{halfpoly}} 
\label{append-halfinteg} 
We give an alternate proof of half-integrality of $\Pc$ based on the integrality of the
intersection of two submodular polyhedra.
Observe that by setting $z_i=2v_i$ for $i\in\F$, and introducing slack variables $s_j$ for
every $j\in D$, the system defining $\Pc$ is equivalent to
\begin{equation}
\hspace*{-5ex}
0\leq z(S)\leq 2r(S) \quad \forall S\sse\F, \qquad
z(G_j)+s_j =2, \ \ z(G_j\sm F'_j)+s_j\leq 1,\ s_j\geq 0 \quad \forall j\in D. 
\label{rep1}
\end{equation}
This in turn is equivalent to
\begin{alignat}{2}
z(S)+s(A) & \leq h_1(S\uplus A) \qquad && \forall S\sse\F,\ A\sse D 
\label{h1} \\ 
z(S)+s(A) & \leq h_2(S\uplus A) \qquad && \forall S\sse\F,\ A\sse D \label{h2} \\
z, s & \geq 0 \label{nonneg} \\
z(G_j)+s_j & = 2 \qquad && \forall j\in D \label{h2base} 
\end{alignat}
\noindent
where $h_1$ and $h_2$ are submodular functions defined over $\F\uplus D$ given by
$h_1(S\uplus A):=2r(S)+|A|$ and
$h_2(S\uplus A):=2|\{j: F'_j\cap S\neq\es\}|
+|\{j: F'_j\cap S=\es\text{ and }(j\in A\text{ or }G_j\cap S\neq\es)\}|$. 
(To see the equivalence, it is clear that constraints \eqref{h1}--\eqref{h2base} include
\eqref{rep1}. Conversely, \eqref{h1} follows by adding the constraints $z(S)\leq 2r(S)$
and $s_j\leq 1$ for all $j\in A$; \eqref{h2} is implied by the sum of constraints
$z(G_j)+s_j=2$ for all $j$ such that $F'_j\cap S\neq\es$, and $z(G_j)+s_j\leq 1$ for all
other $j$ such that $j\in A$ or $G_j\cap S\neq\es$.)
Let $\Qc$ be the polytope defined by \eqref{h1}--\eqref{nonneg}.
Since $h_1$ and $h_2$ are integer submodular functions, $\Qc$ is the intersection of the
submodular polyhedra for $h_1$ and $h_2$, which is known to be integral. 
Also, constraints \eqref{h1}--\eqref{h2base} define a face of $\Qc$. 
Now it is easy to see that an extreme point $v$ of $\Pc$, maps to an extreme point
$(2v,s)$, for a suitably defined $s$, of this face (which must be integral). Hence, $\Pc$ 
has half-integral extreme points.

\section{Inapproximability of matroid-intersection median} \label{mintmed}
We show that the problem of deciding if an instance of matroid-intersection median has a
zero-cost solution is \npcomplete. This implies that no multiplicative approximation
factor is achievable in polytime for this problem unless {\em P}=\np.
The reduction is from the \npcomplete\xspace 
{\em directed Hamiltonian path} problem, wherein we are given a directed graph
$D=(N,A)$, and two nodes $s$, $t$, and we need to determine if there is a simple
(directed) $s\leadsto t$ path spanning all the nodes. The facility-set in the
matroid-intersection median problem is the arc-set $A$, and every node except $t$ is a
client.  
One of the matroids $M$ is the graphic matroid on the undirected version of $D$, that is,
an arc-set is independent if it is acyclic when we ignore the edge directions.  The second
matroid $M_2$ is a partition matroid that enforces that every node other than $s$ has at
most one incoming arc. All facility-costs are 0. We set $c_{ij}=0$ if $i$ is an outgoing
arc of $j$, and $\infty$ otherwise. Notice that this forms a metric since the sets 
$\{i: c_{ij}=0\}$ are disjoint for different clients.

It is easy to see that an $s\leadsto t$ Hamiltonian path translates to a zero-cost
solution to the matroid-intersection median problem. Conversely, if we have a zero-cost
solution to matroid-intersection median, then it must open $|N|-1$ facilities, one for
each client. Hence, the resulting edges must form a (spanning) arborescence rooted at $s$,
and moreover, every node other than $t$ must have an outgoing arc. Thus, the resulting
edges yield an $s\leadsto t$ Hamiltonian path.

\section{Knapsack median: algorithm details and analysis} \label{append-knap}
Recall that we consider a more general version of knapsack median than that considered
in~\cite{KrishnaswamyKNSS11,Kumar12,CharikarL12}. Each facility $i$ has a facility-opening  
cost $f_i$ and a {\em weight} $w_i$, and we have a knapsack constraint 
$\sum_{i\in F}w_i\leq B$ constraining the total weight of open facilities. 
The goal is to minimize the sum of the facility-opening and client-connection costs while
satisfying the knapsack constraint on the set of open facilities.
We may assume that we know the maximum facility-opening cost $\fopt$ of a facility opened
by an optimal solution, so in the sequel we assume that $f_i\leq\fopt,\ w_i\leq B$ for
all facilities $i\in\F$.

Krishnaswamy et al.~\cite{KrishnaswamyKNSS11} showed that the natural LP-relaxation
for knapsack median has a bad integrality gap; this holds even after augmenting the
natural LP with knapsack-cover inequalities. To circumvent this difficulty,
Kumar~\cite{Kumar12} proposed the following lower bound, which we also use. 
Suppose that we have an estimate $\copt$ within a $(1+\e)$-factor of the connection cost
of an optimal solution 
(which we can obtain by enumerating all powers of $(1+\e)$). 
Then, defining $U_j:=\arg\max\{z: \sum_kd_k\max\{0,z-c_{jk}\}\leq\copt\}$,
Kumar argued that the constraint $x_{ij}=0$ if $c_{ij}>U_j$ is valid for the knapsack
median instance. We augment the natural LP-relaxation with these constraints to obtain the
following LP \eqref{knaplp}. 
\begin{alignat}{3}
\min && \quad \sum_i f_iy_i &+ \sum_j\sum_i d_j&&c_{ij}x_{ij} \tag{K-P} 
\label{knaplp} \\
\text{s.t.} && \sum_i x_{ij} & \geq 1 && \forall j \notag \\[-7pt]
&& x_{ij} & \leq y_i && \forall i,j \notag \\
&& \sum_{i} w_iy_i & \leq B \notag \\ 
&& x_{ij},y_i & \geq 0  && \forall i,j; \qquad  x_{ij}=0 \quad \text{if $c_{ij}>U_j$}.\notag 
\end{alignat}

Let $(x,y)$ be an optimal solution to \eqref{knaplp} and $\OPT$ be its value. 
Let $\bC_j=\sum_i c_{ij}x_{ij}$.
Note that if our estimate $\copt$ is correct, then $\OPT$ is at most the optimal value
$\iopt$ for the knapsack median instance. We show that $(x,y)$ can be rounded to an
integer solution of cost $\fopt+4\copt+28\cdot\OPT$. Thus, if consider all possible
choices for $\copt$ in powers of $(1+\e)$ and pick the solution returned with least cost,  
we obtain a solution of cost at most $(32+\e)$ times the optimum. The rounding procedure
is as follows.

\begin{list}{K\arabic{enumi}.}{\usecounter{enumi} \addtolength{\leftmargin}{-1ex}}
\item {\bf Consolidating demands.}
We start by consolidating demands as in Step I in Section~\ref{halfinteg}.
We now work with the client set $D$ and the demands $\{d'_j\}_{j\in D}$. 
For $j\in D$, we use $M_j\sse\D$ to denote the set of clients (including $j$) whose
demands were moved to $j$. Note that the $M_j$s partition $\D$.
Let $\OPT'$ denote the cost of $(x,y)$ for this modified instance. 
As before, for each $j\in D$ we define 
$F_j=\{i:c_{ij}=\min_{k\in D}c_{ik}\}$, $F'_j=\{i\in F_j:c_{ij}\leq 2\bC_j\}$, 
$\gm_j=\min_{i\notin F_j}c_{ij}$, and $G_j=\{i\in F_j:c_{ij}\leq\gm_j\}$.

\medskip
\item {\bf Obtaining a nearly half-integral solution.}
Set $y'_i=x_{ij}\leq y_i$ if $i\in G_j$, and $y'_i=0$ otherwise. 
Let $\F'=\bigcup_{j\in D}G_j$. In the sequel, we will only consider facilities in $\F'$.
Consider the following polytope:
\begin{equation}
\Kc:=\Bigl\{v\in\R_+^{\F'}: v(F'_j)\geq\tfrac{1}{2}, \ \ 
v(G_j)\leq 1 \quad \forall j\in D,
\qquad \sum_i w_iv_i\leq B\Bigr\}. \label{kpoly}
\end{equation}
Define 
$K(v)=\sum_i 2f_iv_i+\sum_j d'_j\bigl(2\sum_{i\in G_j}c_{ij}v_i+8\gm_j(1-v(G_j))\bigr)$ for
$v\in\R_+^{\F'}$. 
Since $y'\in\Kc$, we can efficiently obtain an extreme point $\hy$ of
$\Kc$ such that $K(\hy)\leq K(y')$, the support of $\hy$ is a subset of the
support of $y'$, and all constraints that are tight under $y'$ remain tight under $\hy$.%
\footnote{We can obtain $\hy$ as follows. 
Let $Av\leq b, v\geq 0$ denote the constraints
of $\Kc$. Recall that $z$ is an extreme point of $\Kc$ iff the submatrix $A'$ of $A$
corresponding to the non-zero variables and the tight constraints has full column rank. So
if $y'$ is not an extreme point, then letting $\F''=\{i:y'_i>0\}$, we can find some
$d'\in\R^{\F''}$ such that $A'd=0$. So letting $d_i=d'_i$ if $i\in\F''$ and 0 otherwise,
we can find some $\e>0$ such that both $y'+\e d$ and $y'-\e d$ are feasible and all
constraints that were tight under $y'$ remain tight. So moving in the direction that does
not increase the $K(.)$-value until some non-zero $y'_i$ drops down to 0 or some new
constraint goes tight, and repeating, we obtain the desired extreme point $\hy$.}
Thus, if $i\in G_j$ and $\hy_i>0$, then $y'_i>0$ and so $c_{ij}\leq U_j$. Also, if
$\hy(G_j)<1$ then $y'(G_j)<1$, and so $\gm_j\leq U_j$.
%
We show in Lemma~\ref{knappoly} that there is {\em at most one} client, which we call the
{\em special client} and denote by $s$, such that $G_s$ contains a facility $i$ with
$\hy_i\notin\bigl\{0,\frac{1}{2},1\bigr\}$. 


As in Section~\ref{improved}, for each client $j\in D$, define $\sg(j)=j$ if $\hy(G_j)=1$,
and $\sg(j)=\arg\min_{k\in D: k\neq j}c_{jk}$ otherwise (breaking ties arbitrarily).
Note that $c_{j\sg(j)}\leq 2\gm_j$. 
We now define the primary and secondary facilities of each client $j\in D$, which we
denote by $i_1(j)$ and $i_2(j)$ respectively. 
If $j$ is not the special client $s$, then $i_1(j)$ is the facility $i$ nearest to $j$
with $\hy_i>0$; otherwise, $i_1(j)=\arg\min_{i\in F'_j: \hy_i>0}w_i$ (breaking ties
arbitrarily). 
If $\hy_{i_1(j)}=1$, then we set $i_2(j)=i_1(j)$. If $\hy(G_j)<1$, we set
$i_2(j)=i_1(\sg(j))$. 
If $\hy_{i_1(j)}<\hy(G_j)=1$, we set $i_2(j)$ to: the half-integral facility in $G_j$
other than $i_1(j)$ that is nearest to $j$ if $j\neq s$; and the facility with smallest
weight among the facilities $i\in G_j$ with $\hy_i>0$ (which could be the same as $i_1(j)$)
if $j=s$. 
Define $S_j=\{i_1(j),i_2(j)\}$.

To gain some intuition, observe that the facilities $i_1(j)$ and $i_2(j)$ naturally yield
a half-integral solution, where these facilities are open to an extent of $\frac{1}{2}$
and $j$ is assigned to them to an extent of $\frac{1}{2}$; as before, if $i_1(j)=i_2(j)$,
then this means that $i_1(j)$ is open to an extent of 1 and $j$ is assigned completely to
$i_1(j)$. The choice of the primary and secondary facilities ensures 
that this solution is feasible. (We do not however modify $\hy$ as indicated above.)

\medskip
\item {\bf Clustering and rounding to an integral solution.}
This step is quite straightforward. We define $C'_j$ for $j\in D$, and cluster clients in
$D$ exactly as in step A2 in Section~\ref{improved}, and we open the facility with
smallest weight within each cluster. 
Finally, we assign each client to the nearest open facility. Let $(\tx,\ty)$ denote the
resulting solution.
Recall that $D'$ is the set of cluster centers, and for $k\in D$, $\ctr(k)$ denotes
the client in $D$ due to which $k$ was removed in the clustering process (so $\ctr(j)=j$
for $j\in D'$). 
\end{list}

\paragraph{Analysis.}
We call a facility $i$ half-integral (with respect to the vector $\hy$ obtained in step
K2) if $\hy_i\in\{0,\frac{1}{2},1\}$ and fractional otherwise. 

\begin{lemma} \label{knappoly}
The extreme point $\hy$ of $\Kc$ obtained in step K2 is such that 
there is at most one client, called the {\em special client} and denoted by $s$, such that
$G_s$ contains fractional facilities.
Moreover, if $\frac{1}{2}<\hy(G_s)<1$, then there is one exactly one facility $i\in F'_s$
such that $\hy_i>0$. 
\end{lemma}

\begin{proof}
Since $\hy$ is an extreme point, it is well known 
that the submatrix $A'$ of the constraint matrix whose columns correspond to the non-zero 
$\hy_i$s and rows correspond to the tight constraints under $\hy$ has full column-rank. 
The rows and columns of $A'$ may be accounted for as follows. Each client $j\in D$
contributes:  
(i) a non-empty disjoint set of columns corresponding to the positive $\hy_i$s in
$G_j$; and 
(ii) a possibly-empty disjoint set of at most two rows corresponding to the tight
constraints $\hy(F'_j)=\frac{1}{2}$ and $\hy(G_j)=1$.
This accounts for all columns of $A'$. 
There is at most one remaining row of $A'$, which corresponds to the tight constraint
$\sum_i w_i\hy_i=B$. 

Let $p_j$ and $q_j$ denote respectively the number of columns and rows contributed by
$j\in D$. First, note that $p_j\geq q_j$ for all $j\in D$. This is clearly true if
$q_j\leq 1$; if $q_j=2$, then $\hy(F'_j)=\frac{1}{2},\ \hy(G_j)=1$, so both $F'_j$ and
$G_j$ must have at least one positive $\hy_i$. Also, note that if $p_j=q_j$, then $G_j$
contains only half-integral facilities. Since $\sum_j p_j\leq\sum_j q_j+1$, there can be
at most one client such that $p_j>q_j$; we let this be our special client $s$. Note that
we must have $p_s=q_s+1$.

If $\frac{1}{2}<\hy(G_s)<1$ then: (i) $q_s=0$, so $p_s=1$; or (ii) $q_s=1$, so
$p_s=2$, and since $\hy(F'_s)=\frac{1}{2}<\hy(G_s)$, both $F'_s$ and $G_s$ contain exactly
one positive $\hy_i$.
\end{proof}

It is easy to adapt the proof of Lemma~\ref{half}, and obtain that 
$K(\hy)\leq K(y')\leq 8\cdot\OPT'\leq 8\cdot\OPT$.
Next, we prove our main result: the integer solution $(\tx,\ty)$ computed is feasible and
its cost for the modified instance is at most 
$K(\hy)+\fopt+4\copt+16\cdot\OPT$. Thus, ``moving'' the consolidated demands back to their
original locations yields a solution of cost at most $(32+\e)\cdot\iopt$ for the correct
guess of $\fopt$ and $\copt$. The following claims will be useful.

\begin{claim} \label{knapclus}
If $\hy(G_j)=1$ for some $j\in D$, then (we may assume that) $j$ is a cluster center.
\end{claim}

\begin{proof}
Let $i'=i_1(j),\ i''=i_2(j)$. Let $k\in D$ be such that $S_k\cap S_j\neq\es$. Then
$\sg(k)=j$. So 
$2(C'_k-C'_j)=c_{i_1(k)k}+c_{jk}-c_{i_2(j)j}\geq c_{i_1(k)j}-c_{i_2(j)j}\geq 0$ since
$i_2(k)\notin G_j$. 
\end{proof}

\begin{claim} \label{knapclaim}
For any client $j\in D$, we have $d'_jU_j\leq\copt+4\cdot\OPT$.
\end{claim}

\begin{proofnobox}
By definition, $\sum_kd_k\max\{0,U_j-c_{jk}\}\leq\copt$. 
So $d'_jU_j=\sum_{k\in M_j}d_kU_j$, which equals
\begin{equation*}
\sum_{k\in M_j}d_k(U_j-c_{jk})+\sum_{k\in M_j}d_kc_{jk}
\leq\copt+\sum_{k\in M_j}4d_k\bC_k\leq\copt+4\cdot\OPT. \tag*{\qedsymbol} 
\end{equation*}
\end{proofnobox}

\vspace{-3ex}

\begin{theorem} \label{knapthm}
The solution $(\tx,\ty)$ computed in step K3 for the modified instance is feasible and has
cost at most $K(\hy)+\fopt+4\copt+16\cdot\OPT$.
\end{theorem}

\begin{proof}
Let $B_j(v)=d'_j\bigl(2\sum_{i\in G_j}c_{ij}v_i+8\gm_j(1-v(G_j))$ for $v\in\R_+^{\F'}$. 
So $K(\hy)=2\sum_i f_i\hy_i+\sum_j B_j(\hy)$.
Recall that $S_j=\{i_1(j),i_2(j)\}$ for every $j\in D$.

We first prove feasibility and bound the total facility-opening cost.
Consider a cluster centered at $j$. Let $i'=i_1(j),\ i''=i_2(j)$. Let $\hi$ be the 
facility opened from $S_j$. 
If $\hy(S_j)=1$, then $w_{\hi}\leq\sum_{i\in S_j}w_i\hy_i$. Otherwise, either $j=s$ or
$\sg(j)=s$. If $j=\sg(j)=s$, then $\hi$ is the least-weight facility in $G_j$. Otherwise,
if $j=s$ then $\hi$ is the least-weight facility in $F'_j\cup\{i_2(j)\}$ and
$\hy(F'_j)+\hy_{i_2(j)}\geq 1$; finally, if $j\neq\sg(j)=s$ then $\hi$ is the least-weight
facility in $\{i_1(j)\}\cup F'_{\sg(j)}$ and $\hy_{i_1(j)}+\hy(F'_{\sg(j)})\geq 1$. 
Since $S_j\sse G_j\cup G_{\sg(j)}$, in every case, we have 
$w_{\hi}\leq\sum_{i\in G_j\cup G_{\sg(j)}}w_i\hy_i$. 

If all facilities in $S_j$ are half-integral, then 
$f_{\hi}\leq 2\sum_{i\in S_j}f_i\hy_i\leq 2\sum_{i\in G_j\cup G_{\sg(j)}}f_i\hy_i$. 
Otherwise, we have $j=s$ or $\sg(j)=s$, and we bound $f_{\hi}$ by $\fopt$.

Note that if $k\in D'$ is some other cluster center,
then $G_j\cup G_{\sg(j)}$ is disjoint from $G_k\cup G_{\sg(k)}$. If not, then we must have
$\sg(j)=k$ or $\sg(k)=j$ or $\sg(j)=\sg(k)$, 
which yields the contradiction that $S_j\cap S_k\neq\es$.
So summing over all clusters, we obtain that the total weight of open facilities is at most 
$\sum_{j\in D'}\sum_{i\in G_j\cup G_{\sg(j)}}w_i\hy_i\leq\sum_iw_i\hy_i\leq B$, 
and the facility opening cost is at most $2\sum_if_i\hy_i+\fopt$.

\medskip
We now bound the total client-assignment cost. Fix a client $j\in D'$. 
The assignment cost of $j$ is at most $d'_jc_{i_2(j)j}$. Note that $c_{i_2(j)j}\leq 3U_j$.  
If $j\neq s$, then $B_j(\hy)\geq d'_jc_{i_2(j)j}$: this holds if $\hy(G_j)=1$ since
$\hy_{i_2(j)}\geq\frac{1}{2}$; otherwise, $B_j(\hy)\geq 4d'_j\gm_j\geq d'_jc_{i_2(j)j}$.  
If $j=s$, then its assignment cost is at most $3d'_jU_j\leq 3\copt+12\cdot\OPT$
(Claim~\ref{knapclaim}).

Now consider $k\in D\sm D'$. Let $j=\ctr(k)$, and $i'=i_1(j),\ i''=i_2(j)$. 
We consider two cases. 
\begin{list}{\arabic{enumi}.}{\usecounter{enumi} \topsep=0.5ex \itemsep=0.5ex
    \addtolength{\leftmargin}{-1.5ex}} 
\item $i_1(k)\in S_j$. Then $k=\sg(j)$ and $k$'s assignment cost is at most
$d'_kc_{i_2(k)k}$. As above, this is bounded by $B_k(\hy)$ if $k\neq s$, and by
$3\copt+12\cdot\OPT$ otherwise.

\item $i_1(k)\notin S_j$. Let $\ell=\sg(k)$. 
We claim that the assignment cost of $k$ is at most $d'_k\bigl(c_{i_1(k)k}+4\gm_k\bigr)$.
To see this, first suppose $\ell\neq j$, and so $\ell=\sg(j)$. Then, $k$'s assignment cost
is at most $d'_k\bigl(c_{k\ell}+c_{\ell j}+c_{i'j}\bigr)\leq d'_k\bigl(2c_{k\ell}+c_{i_1(k)k}\bigr)
\leq d'_k\bigl(c_{i_1(k)k}+4\gm_k\bigr)$, where the first inequality follows since  
$C'_j\leq C'_k$. If $\ell=j$, then $i_2(k)=i_1(j)=i'$ and $k$'s assignment cost is at most  
$d'_k\bigl(c_{jk}+c_{j\sg(j)}+c_{i_2(j)\sg(j)}\bigr)
\leq d'_k\bigl(c_{i_1(k)}+2c_{jk}\bigr)\leq d'_k\bigl(c_{i_1(k)k}+4\gm_k\bigr)$, where the
first inequality again follows from $C'_j\leq C'_k$.

Since $k\notin D'$, we have $\hy(G_k)<1$ (by Claim~\ref{knapclus}). So $y'(G_k)<1$ and
$\gm_k\leq U_k$. If $k\neq s$, then $B_k(\hy)\geq d'_k\bigl(c_{i_1(k)k}+4\gm_k\bigr)$.
If $k=s$ and $\hy(G_k)=\frac{1}{2}$, then $B_k(\hy)\geq 4d'_k\gm_k$ and 
$d'_kc_{i_1(k)k}\leq d'_kU_k$. Otherwise, by Lemma~\ref{knappoly}, we have
$\hy_{i_1(k)}>\frac{1}{2}$, and so $B_k(\hy)\geq d'_kc_{i_1(k)k}$ and 
$4d'_k\gm_k\leq 4d'_kU_k$. Taking all cases into account, we can bound $k$'s assignment
cost by $B_k(\hy)$ if $k\neq s$, and by 
$B_k(\hy)+4d'_kU_k\leq B_k(\hy)+4\copt+16\cdot\OPT$ if $k=s$.
\end{list}

\smallskip \noindent
Putting everything together, the total cost of $(\tx,\ty)$ is at most 
$2\sum_i f_i\hy_i+\sum_j B_j(\hy)+\fopt+4\copt+16\cdot\OPT
=K(\hy)+\fopt+4\copt+16\cdot\OPT$. 
\end{proof}

\begin{corollary} \label{knapcor}
There is a $(32+\e)$-approximation algorithm for the knapsack median problem.
\end{corollary}

\end{document}